\documentclass[11 pt]{article}
\date{}

\usepackage{graphicx,subfigure}
\usepackage{amsmath,amssymb}
\usepackage{mystyle}
\usepackage{color}
\usepackage{graphicx}

\begin{document}

\title{Information Exchange Limits in Cooperative MIMO Networks}

\author{Ali Tajer\thanks{Electrical Engineering Department, Princeton University, Princeton, NJ 08540 (email: tajer@princeton.edu).} \and Xiaodong Wang \thanks{Electrical Engineering Department, Columbia University, New York,
NY 10027 (email: wangx@ee.columbia.edu).}}

%\author{Ali Tajer, \IEEEmembership{Member, IEEE,}\thanks{Ali Tajer is with the Electrical Engineering Department, Princeton University, Princeton, NJ 08544 (tajer@princeton.edu).} and Xiaodong Wang, \IEEEmembership{Fellow, IEEE}
%\thanks{Xiaodong Wang is with the Electrical Engineering Department, Columbia University, New York, NY 10027 (wangx@ee.columbia.edu).}}

\maketitle

\begin{abstract}
Concurrent presence of {\em inter}-cell and {\em intra}-cell interferences constitutes a major impediment to reliable downlink transmission in multi-cell multiuser networks. Harnessing such interferences largely hinges on two levels of information exchange in the network: one from the users to the base-stations ({\em feedback}) and the other one among the base-stations ({\em cooperation}). We demonstrate that exchanging a {\em finite} number of bits across the network, in the form of feedback and cooperation, is adequate for achieving the optimal capacity scaling. We also show that the {\em average} level of information exchange is independent of the number of users in the network. This level of information exchange is considerably less than that required by the existing coordination strategies which necessitate exchanging {\em infinite} bits across the network for achieving the optimal sum-rate capacity scaling. The results provided rely on a constructive proof.
\end{abstract}

%\begin{keywords}
%Cooperative MIMO, coordination, feedback, finite-rate, capacity scaling law.
%\end{keywords}

\section{Introduction}
\label{sec:intro}

Multiple-antenna broadcast transmission is known to offer substantial improvement in spectral efficiency via spatial multiplexing. Such gains depend heavily on the availability of some channel state information (CSI) at the transmitter. Acquiring such CSI allows for managing the intra-cell interference which is due to spatial multiplexing. There exists a rich literature on the design of user scheduling that spans a wide range of transceiver designs with different levels of complexity and CSI feedback requirements \cite{Caire:IT03, Goldsmith:JSAC03, Sharif:COM07}.

In {\em multi}-cell multiuser networks, besides the intra-cell interference caused by the residual interference remaining after spatial multiplexing within each cell, another impediment arises from inter-cell interference due to the ever-shrinking cell sizes, which makes the networks operate in interference-limited, rather than noise-limited, regimes. Alleviating the effects of inter-cell interference requires the base-stations to cooperatively adjust their transmissions~\cite{Gesbert:SPM07, Zhang:WCOM08, Huang:WCOM09, Zhang:WCOM09}. Coordination among the base-stations for downlink transmission imposes some information exchange (cooperation) among the base-stations. Once {\em full} CSI feedback and {\em full} cooperation is viable, i.e., {\em all} the users report their CSI to {\em all} base-stations, and {\em all} base-stations share their information streams, the multi-cell network can be essentially modeled as a single-cell multiple-input multiple-output (MIMO) broadcast channel with a transmit antenna array of a higher dimension. Using the literature on MIMO broadcast channel immediately establishes the optimal sum-rate capacity scaling of multi-cell multiuser networks which is achievable via dirty-paper coding~\cite{Sharif:COM07}.

Such optimal capacity scaling, however, necessitates {\em full} feedback and cooperation across the network. This requirement, specially in large networks, can potentially overburden the feedback channels as well as the backbone channels that connect the base-stations.  Therefore, it becomes imperative to study whether such full information exchange is all necessary or contains some redundancy that can be avoided. The question we are specifically trying to address is: what is the minimum amount of information exchange, in the form of either CSI feedback or base-station cooperation, that ensures retaining the optimal sum-rate capacity scaling?

We consider a multi-cell network model, in which every few neighboring cells are collectively regarded as a super-cell that serve the users within that super-cell collaboratively. The analysis demonstrates that for the purpose of achieving the optimal scaling, cooperation among the super-cells, while desirable, is not necessary. In other words, while collective transmission by the base-stations across the super-cells has certain advantages, as long as achieving the optimal capacity scaling is concerned, the super-cells can decide about their scheduling policies independently. Furthermore, each user has to provide feedback to only the base-stations of its designated super-cell. These two relaxations essentially imply that base-stations of each super-cell collect only {\em local} information and any information exchange across super-cells are avoided. We also quantify the aggregate amount of information exchange necessary for achieving the optimal sum-rate capacity scaling. Specifically, we show that when the network contains $M$ super-cells, each consisting of $Q$ base-stations and when each base-station is equipped with $N_t$ transmit antennas, the minimum aggregate amount of information exchange in the network to guarantee optimal capacity scaling is upper bounded by $MQN_t \lceil\log QN_t\rceil $, which is independent of the number of users in the network.

These results in this paper are provided through a constructive proof. The flow of the proof is as follows. After discussing the network model in Section~\ref{sec:model} we describe a scheduling procedure and feedback mechanism in Section~\ref{sec:algorithm}. Next, in Section~\ref{sec:analysis} we show that proposed scheduler yields the optimal sum-rate capacity scaling. Then  in Section~\ref{sec:feedback} we assess the aggregate amount of information exchange that the proposed scheduler requires.  Section~\ref{sec:fairness} addresses some practical issues, e.g., the impact of the size of the super-cells and fairness, and finally Section~\ref{sec:conclusion} concludes the paper. To enhance the flow of the content, most of the proofs are relegated in the appendices.

We remark that as opposed to multi-cell MIMO networks, there exists a rich literature on single-cell downlink transmission that investigate the interplay between the feedback load and broadcast sum-rate throughput scaling. One group of studies, c.f. \cite{Sharif:IT05, Yoo:JSAC07, Huang:SP07}, for single-cell systems consider feeding back \emph{unquantized} real values (channel magnitude or $\snr$, which in theory requires infinite-rate feedback.  In another line of studies, c.f.,~\cite{Huang:VT09}, finite-rate feedback mechanisms are introduced. However, the optimal sum-rate throughput is achievable in these studies when the number of feedback bits tends to infinity. The proposed feedback mechanisms in~\cite{Sanayei:IT07, Sanayei:WCOM07}, impose only finite-rate feedback loads and achieve the optimal multiuser diversity gain. However, since they do not allow for simultaneous transmissions to multiple users, they cannot capture the multiplexing gain. Finally, the heuristic feedback mechanisms proposed in \cite{Tang:COML:05, Tang:COM:07} have limited feedback loads for which the sum-rate throughput scaling analysis is not provided. As a by-product of our analysis we find that in a {\em single}-cell downlink transmission it is possible to achieve the optimal capacity scaling rate with {\em finite}-rate feedback. It is also noteworthy on a different research direction, that capacity scaling for some infrastructure-based sensor networks are provided in \cite{Shin:arxive08}, \cite{Liu:INFOCOM08}

\section{System Model}
\label{sec:model}
Consider the downlink transmission in a MIMO network that consists of $M$ super-cells, each comprising of $Q$ regular cells equipped with one base-station. We use the shorthand $\Q_m$ to refer to the $m^{th}$ super-cell and denote its base-stations by $B^r_m,$ for $r\in\{1,\dots,Q\}$. Let $K$ denote the number of users per super-cell and assume that the base-stations of each super-cell accomplish the downlink transmissions collaboratively. Moreover, we assume that all the base-stations within a super-cell are fully cooperative, i.e., are connected through a backbone link that allows them to share their information streams and to carry out fully coordinated downlink transmissions. The case that such cooperation is not viable corresponds to the setting $Q=1$. We also use the shorthand $U^k_n$ to denote the $k^{th}$ user in the $n^{th}$ super-cell, respectively, for $k\in\{1,\dots,K\}$ and $n\in\{1,\dots,M\}$.

Define  $\bH_{n,m}^{k,r}\in\mathbb{C}^{N_r\times N_t}$ as the downlink channel from base-station $B^r_m$ to user $U^k_n$ and let $\bx^r_m(t)\in\mathbb{C}^{N_t\times 1}$ denote the signal vector transmitted by base-station $B^r_m$ at time instance $t$. The received signal by $U^k_n$ at time instance $t$ is therefore given by
\begin{equation*}
    \by^k_n(t)=\sum_{m=1}^M\sum_{r=1}^Q\sqrt{\gamma^{k,r}_{n,m}}\ \bH^{k,r}_{n,m}\bx^r_m(t)+\bz^k_n(t)\ ,
\end{equation*}
where $\gamma^{k,r}_{n,m}\in\mathbb{R}^+$ is incorporated to capitalize the combined effect of transmission power and signal attenuation (path-loss) between $B^r_m$ and $U^k_n$. In non-homogenous networks the terms $\{\gamma^{k,r}_{n,m}\}$ are distinct as different users undergo different path-loss and shadowing effects. Also,  $\bz^k_n\in\mathbb{C}^{N_r\times 1}$ accounts for the additive white Gaussian noise distributed as $\mathcal{CN}(\boldsymbol{0},\bI_{N_r})$. We consider block frequency-flat fading channels and assume that the entries of the channel matrix $\bH^{k,r}_{n,m}$ are independent and identically distributed as complex Gaussian $\mathcal{CN}(0,1)$. In order to facilitate spatial multiplexing and to harness its ensuing interference as well as the inter-cell interference, the base-stations employ linear precoding. We define $\{\bomega^r_m(\ell)\}_{\ell=1}^{N_t}$ as the set of linear precoders employed by base-station $B^r_m$. By further defining  $x^r_m(\ell,t)$ as the information symbol that $B^r_m$ transmits along the beamforming vector $\bomega^r_m(\ell)$, at any time instance $t$ we have
\begin{equation}
    \label{eq:transmit}
    \bx^r_m(t)=\sum_{\ell=1}^{N_t}x^r_m(\ell,t)\;\bomega^r_m(\ell)\ .
\end{equation}
Motivated by refraining from information feedback about channel directions the base-stations exploit the notion of orthogonal random beamforming first proposed in~\cite{Sharif:IT05}. For this purpose, at the beginning of each block transmission, each base-station generates  $N_t$ orthonormal vectors $\{\bomega^r_m(\ell)\}_{\ell=1}^{N_t}$, according to an isotropic distribution~\cite{Hassibi:IT02}. The user $U^k_n$ knows $\bH^{k,r}_{n,n}\bomega^r_n(\ell)$ for $\ell\in\{1,\dots,N_t\}$ and $r\in\{1,\dots,Q\}$ perfectly and instantaneously. We denote the transmission power per information stream by $\rho\dff\frac{P}{N_t}$, where $P$ is the average power per base-station. Also, throughout the paper, we use the convention $a_K\doteq b_K$ to denote asymptotic equality, i.e., $\lim_{K\rightarrow\infty}\frac{a_K}{b_K}=1$. Operators $\dotlt$ and $\dotgt$ are defined accordingly. All the logarithms, unless otherwise mentioned, are in base 2 and throughput are in bits/sec/Hz.

\section{Scheduling Framework}
\label{sec:algorithm}
We first consider the case that the users have one receive antenna and will later generalize it the case of arbitrary number of receive antennas in Section~\ref{sec:analysis}. We assume that the users employ single-user decoders such that all the inter-cell and intra-cell interferers are treated as Gaussian noise. While more advanced receivers will offer better performance, we show that the simple single-user decoders are sufficient to fully capture the optimal multiplexing and multiuser diversity gains in MIMO networks.

Due to full cooperation, the information stream that some base-station in the super-cell $\Q_n$ intends to transmit to user $U^k_n$ is known to all base-stations of that super-cell. Therefore, the message intended for user $U^k_n$ can be transmitted along any arbitrary beamforming vector of any base-station of $\Q_n$. Hence, user $U^k_n$ can compute the following $Q\times N_t$ distinct $\sinr$s given that its intended information is embedded in $x^r_n(\ell,t)$ and transmitted by the base-station $B^r_m$ along $\bomega^r_n(\ell)$, where $r\in\{1,\dots,Q\}$ and $\ell\in\{1,\dots,N_t\}$,
\begin{align}\label{eq:sinr}
    \sinr^k_n(r,\ell)= \frac{\gamma_{n,n}^{k,r}|\bH^{k,r}_{n,n}\bomega^r_n(\ell)|^2} {\sum_{m=1}^M\sum_{r'=1}^Q\gamma^{k,r'}_{n,m}\sum_{\ell'=1}^{N_t}|\bH^{k,r'}_{n,m}\bomega^{r'}_m(\ell')|^2 -\gamma_{n,n}^{k,r}|\bH^{k,r}_{n,n}\bomega^{r}_n(\ell)|^2+\frac{1}{\rho}}\ .
\end{align}
It can be readily verified that for any given $\ell$ and for any $(k,n,r)\neq (k',n',r')$, the numerators and denumerators of ${\sf SINR}^k_n(r,\ell)$ and ${\sf SINR}^{k'}_{n'}(r',\ell)$ are disjoint. In other words, any term of the form $\bH^{k,r}_{n,m}\bomega^r_n(\ell)$ that appears in the characterization of ${\sf SINR}^k_n(r,\ell)$ does not appear in that of ${\sf SINR}^{k'}_{n'}(r',\ell)$. By invoking the statistical independence of the channel coefficients and beamforming vectors, we find that ${\sf SINR}^k_n(r,\ell)$ and ${\sf SINR}^{k'}_{n'}(r',\ell)$ are also statistically independent  for any $(k,n,r)\neq (k',n',r')$. However, by noting the different attenuation factors $\{\gamma^{k,r}_{n,m}\}$ for different choices of $k,n,r$, ${\sf SINR}^k_n(r,\ell)$ and ${\sf SINR}^{k'}_{n'}(r',\ell)$ are not identically distributed. Similarly it can be shown that for any fixed set of $k,n,r$  the terms $\{{\sf SINR}^k_n(r,\ell)\}$ are statistically identical but not independent. We denote the cumulative distribution function (CDF) of $\sinr^k_n(r,\ell)$ for $\ell\in\{1,\dots,N_t\}$ by
\begin{equation*}
    F^k_{n}(x;r):\mathbb{R}_+\rightarrow[0,1]\ .
\end{equation*}
We propose a user scheduling procedure that aims at identifying and serving $QN_t$ users in each super-cell such that 1) the network sum-rate throughput captures the full multiuser diversity gain, 2) does not impose any information exchange or cooperation among the super-cells, and 3) necessitates each base-station to acquire only a {\em finite} number of feedback bits from each of its designated users. The steps involved in the proposed scheduling procedure are summarized as follows. We remark that the notion of random beamforming, first deployed in \cite{Sharif:IT05}, has been the core of many scheduling schemes for single-cell broadcast channels.
\begin{description}
  \item[1) Normalizing $\sinr$s:] For each user $U^k_n$ we define $Q$ positive scalars $\{\beta^k_n(1),\dots,\beta^k_n(Q)\}$  each pertaining to one of the base-station of the super-cell $\Q_n$.
      \begin{equation}\label{eq:lambda}
        \forall r\in\{1,\dots,Q\}:\quad\mbox{define} \quad\beta^k_{n}(r)\;\;\mbox{such that}\;\;F^k_{n}\big(\beta^k_{n}(r);r\big)=1-\frac{1}{K}\ .
      \end{equation}
      There exist {\em unique} solutions for $\beta^k_{n}(r)$ as $F^k_{n}(x;r)$ is a CDF and strictly monotonic in $x$. Next, corresponding to each user  $U^k_n$ we define the following normalized $\sinr$ terms.
      \begin{equation}\label{eq:sinrn}
        \sinrn^k_{n}(r,\ell)\dff\frac{\sinr^k_{n}(r,\ell)}{\beta^k_{n}(r)}\ ,\quad\forall r\in\{1,\dots,Q\}\quad\mbox{and}\quad \ell\in\{1,\dots,N_t\}\ .
      \end{equation}

  \item[2) Feedback:]\textcolor{white}{a}
  \begin{enumerate}
  \item Recall that each user $U^k_n$ can be served along any beamformer of the super-cell $\Q_n$. Among all beamformers, we identify the most favorable beam for the user $U^k_n$, that is the beamformer yielding the highest $\sinr$ for $U^k_n$, i.e.,
      \begin{equation}\label{eq:sinrbar}
        \bsinr^k_{n}\dff \max_{r,\;\ell}\sinrn^k_{n}(r,\ell)\ .
      \end{equation}
  \item If $\bsinr^k_{n}\geq 1$, then user $U^k_n$ feeds back the index of its most favorable beam to the base-stations of the super-cell $\Q_n$. As a super-cell has $QN_t$ beamformers, conveying such feedback requires $ \lceil\log QN_t\rceil $ information bits.
  \end{enumerate}
  \item \item[3) User Selection:]\textcolor{white}{a}
  \begin{enumerate}
  \item Based on the information fed back by the users, corresponding to each super-cell $\Q_n$ we construct the sets $\CH^r_n(\ell)$, for $r\in\{1,\dots,Q\}$ and $\ell\in\{1,\dots,N_t\}$, such that $\CH^r_n(\ell)$ contains the indices of the users in $\Q_n$ who have declared $\bomega^r_n(\ell)$ as their most favorable beam. By invoking \eqref{eq:sinrbar} we have
      \begin{equation}
        \label{eq:H}
        \CH^r_n(\ell)\dff\Big\{k\med \bsinr^k_{n}\geq 1\quad\mbox{and}\quad \bsinr^k_{n}= \sinrn^k_{n}(r,\ell) \Big\}\ .
      \end{equation}
  \item Super-cell $\Q_n$ {\em randomly} selects one user from each set $\CH^r_n(\ell)$ for $r\in\{1,\dots,Q\}$ and $\ell\in\{1,\dots,N_t\}$ and embeds the information signal intended for this user in the stream $x^r_n(\ell)$. Note that the sets $\{\CH^r_n(\ell)\}_{r,\ell}$ are mutually disjoint and as a result, it is guaranteed that a user will not be served by more than one beam.
\end{enumerate}
\end{description}
This procedure involves information exchange only in the form of feedback from the users to their designated base-stations and the super-cells do not cooperate or exchange any information. The crucial part of such scheduling is the appropriate design of the normalization factors $\beta^k_n(r)$ that results in a user selection yielding the optimal capacity scaling. The choices of the normalization factors as characterized in \eqref{eq:lambda} serve a two-fold purpose, namely they control the aggregate amount of feedback and maintain fairness. As the network size $K$ increases, the normalization factors $\beta^k_n(r)$ as defined in \eqref{eq:lambda} affect the feedback load in two ways. On one hand more users are competing for accessing the resources and it is expected that more users participate in the feedback process. On the other hand, as $K$ increases the normalization factors increase, which makes it harder for individual users in a network to satisfy the conditions for feeding back information to their designated  base stations. As shown later in the paper, the combined impact of these two effects is in favor of reducing the feedback load. Also, as for the fairness, incorporating the statistics of the $\sinr$s in the choices of the normalization factors ensures that all the users are equiprobable in satisfying their feedback constraint and consequently by invoking the selection procedure, all are equiprobable in having access to the resources. The proposed scheduling mechanism requires the users acquire the $\sinr$s yielded by different beamformers. Such $\sinr$ acquisition can be accomplished during the training phase when the users obtain the channel states prior to the downlink transmission of data.

\section{Throughput Scaling Analysis}
\label{sec:analysis}
When in a MIMO network perfect cooperation (full information exchange) among all super-cells is allowed, it can be effectively modeled as a large single-cell multiple-antenna broadcast channel with $MQN_t$ transmit antennas and $MK$ users. In ~\cite{Sharif:COM07, Sharif:IT05} it shown that for such downlink transmission the sum-rate capacity, which is achievable via dirty-paper coding, scales with increasing $K$ as
\begin{equation}\label{eq:DPC_scale}
    \bbe_\bH[R]\doteq MQN_t\log\log KN_r\ .
\end{equation}
In the sequel we analyze the sum-rate throughput scaling of our proposed procedure and demonstrate that the achievable throughput scales similar to \eqref{eq:DPC_scale}.

\subsection{Throughput Scaling for $N_r=1$}
By defining $R_n$ as the sum-rate throughput achieved for the $n^{th}$ super-cell, the expected sum-rate throughput of the network  is $\bbe_\bH[R]=\sum_{n=1}^M\bbe_\bH[R_n]$. We further define $R^r_n(\ell)$ as the rate supported through the transmission along the beamforming vector $\bomega^r_n(\ell)$. As formalized in \eqref{eq:H}, the super-cell $\Q_n$ constructs the new set $\bar\CH_n$, as the set of the user index sets $\CH^r_n(\ell)$ for all $r$ and $\ell$, i.e., $\bar\CH_n\dff\{\CH^r_n(\ell)\}_{r,\ell}$. Conditioned on the sets $\bar\CH_n$, the sum-rate throughput supported in the $n^{th} $ super-cell is the summation of the rates its individual beamformers can sustain. Therefore,
\begin{equation}\label{eq:Rm}
    \bbe_\bH[R_n\med \bar\CH_n]=\sum_{r=1}^Q\sum_{\ell=1}^{N_t}\bbe_\bH[R^r_n(\ell)\med \CH^r_n(\ell)]\ .
\end{equation}
As designated by the user scheduling procedure, for each $r\in\{1,\dots,Q\}$ and $\ell\in\{1,\dots,N_t\}$, one of the members of the set $\CH^r_n(\ell)$ is opted {\em randomly} to be served along the beamforming vector $\bomega^r_n(\ell)$. Hence, each user with its index in $\CH^r_n(\ell)$ will be selected to be served along the beamformer  $\bomega^r_n(\ell)$ with probability $\frac{1}{|\CH^r_n(\ell)|}$. As a result given the set of candidate users $\CH^r_n(\ell)$, the expected value of the rate sustained by the beamformer $\bomega^r_n(\ell)$ is equal to the arithmetic average of the rates these candidate users can sustain. By noting that if the transmission of the beamformer $\bomega^t_n(\ell)$ is intended for user $U^k_n$ then $R^r_n(\ell)=\log(1+\sinr^k_n(r,\ell))$, we obtain
\begin{eqnarray}\label{eq:Rlm}
    \bbe_\bH[R^r_n(\ell)\med \CH^r_n(\ell)]= \frac{1}{|\CH^r_n(\ell)|}\sum_{k\in\CH^r_n(\ell)} \bbe_\bH\left[\log\left(1+\sinr^k_n(r,\ell)\right)\;\big |\;\CH^r_n(\ell)\right].
\end{eqnarray}
The next lemma, which provides a lower bound on $\bbe_\bH[R^r_n(\ell)\med \CH^r_n(\ell)]$ is instrumental to analyzing the scaling rate of $\bbe_\bH[R^r_n(\ell)\med \CH^r_n(\ell)]$.
\begin{lemma}[Conditional Expected Rates]
\label{lemma:sets}
If $\beta^k_n(r)\geq 1$ $\forall n\in\{1,\dots,M\}$, $\forall k\in\{1,\dots,K\}$, and $\forall r\in\{1,\dots,Q\}$, then
\begin{align}\label{eq:Rm_bound}
   \bbe_\bH[R^r_n(\ell)\med \CH^r_n(\ell)]\;\;\geq\;\;\R^r_n(\ell,|\CH^r_n(\ell)|)\;\;\dff\;\;\frac{1}{|\CH^r_n(\ell)|}
   \sum_{j=1}^{|\CH^r_n(\ell)|} \bbe_\bH\left[\log\left(1+\sinr^{(j)}_n(r,\ell)\right)\right]\ ,
\end{align}
where $\sinr^{(j)}_n(r,\ell)$ denotes the $j^{th}$ largest element of the set $\{\sinr^1_n(r,\ell),\dots,\sinr^K_n(r,\ell)\}$.
\end{lemma}
\begin{proof}
See Appendix \ref{app:lemma:sets}.
\end{proof}
Equations \eqref{eq:Rm} and \eqref{eq:Rlm} along with Lemma~\ref{lemma:sets} (conditional expected rates) establish that if $\forall k,n,r$ the normalizing factors satisfy the constraints $\beta^k_n(r)\geq 1$, then we have
\begin{eqnarray}
    \nonumber \bbe_\bH[R_n]&=&\sum_{\bar\CH_n}P(\bar\CH_n) \; \bbe_\bH[R_n\med \bar\CH_n]\\
    \nonumber &\overset{\eqref{eq:Rm}}{=} & \sum_{\bar\CH_n}P(\bar\CH_n)\sum_{r=1}^Q\sum_{\ell=1}^{N_t}\bbe_\bH[R^r_n(\ell)\med \CH^r_n(\ell)] \\
    \label{eq:Rm2_1} &\overset{\eqref{eq:Rm_bound}}{\geq} & \sum_{\bar\CH_n}P(\bar\CH_n)\sum_{r,\;\ell} \R^r_n(\ell,|\CH^r_n(\ell)|) \ .
\end{eqnarray}
Note that the rates $\R^r_n(\ell,|\CH^r_n(\ell)|)$ are identical for all permutations of the set $\CH^r_n(\ell)$ that have the same cardinality. In other words, $\R^r_n(\ell,|\CH^r_n(\ell)|)$ depends on $\CH^r_n(\ell)$ only through its cardinality which can be any member of $\{0,\dots,K\}$. Note that \eqref{eq:Rm2_1} is a weighted sum of the rates $\R^r_n(\ell,b)$ for $\ell\in\{1,\dots,N_t\}$, $r\in\{1,\dots,Q\}$, and $b\in\{1,\dots,K\}$, and in order to further simplify \eqref{eq:Rm2_1} we need to find their weighting coefficients. For this purpose we rearrange all the relevant terms that contribute to $\R^r_n(\ell,k)$ and find that
\begin{equation}\label{eq:Rm2_2}
    \mbox{coefficient of }\; \R^r_n(\ell,b)= \sum_{\{\bar\CH_n\med |\CH^r_n(\ell)|=b\}}P(\bar\CH_n)=P(|\CH^r_n(\ell)|=b)\ .
\end{equation}
Therefore, combining \eqref{eq:Rm2_1} and \eqref{eq:Rm2_2} yields that if $\forall k,n,r$ the constraints $\beta^k_n(r)\geq 1$ are satisfied, then
\begin{eqnarray}
    \nonumber  \bbe_\bH[R_n]& \geq & \sum_{r,\;\ell,\;b} P(|\CH^r_n(\ell)|=b)\; \R^r_n(\ell,b)\\
    \nonumber& \overset{\eqref{eq:Rm_bound}}{=} & \sum_{r,\;\ell,\;b}\frac{P(|\CH^r_n(\ell)|=b)}{b}
    \sum_{j=1}^{b} \bbe_\bH\left[\log\left(1+\sinr^{(j)}_n(r,\ell)\right)\right]\\
    \label{eq:Rm2_22} &= & \sum_{r,\;\ell,\;j} \;\bbe_\bH\left[\log\left(1+\sinr^{(j)}_n(r,\ell)\right)\right] \underset{\dff \; q^j_n(r,\ell)}{\underbrace{\sum_{b=j}^K \frac{P(|\CH^r_n(\ell)|=k)}{k}}}\\
    \label{eq:Rm2_3}&=& \sum_{r,\;\ell,\;j}\; q^j_n(r,\ell)\; \bbe_\bH\left[\log\left(1+\sinr^{(j)}_n(r,\ell)\right)\right] \dff \bar R_n \ .
\end{eqnarray}
By setting $q^0_n(r,\ell)\dff P(|\CH^r_n(\ell)|=0)$ we find that for any fixed $n,r,\ell$ the sequence $\{q^1_n(r,\ell),\dots,q^K_n(r,\ell)\}$ constitutes a probability mass function (PMF) as
\begin{align*}
    \sum_{j=0}^K\; q^j_n(r,\ell) & =q^0_n(r,\ell)+\sum_{j=1}^K\sum_{k=j}^K \;\frac{P(|\CH^r_n(\ell)|=k)}{k}\\
    & = q^0_n(r,\ell)+\sum_{k=1}^K\frac{1}{k}\sum_{j=1}^{k}P(|\CH^r_n(\ell)|=k)= \sum_{k=0}^K\;P(|\CH^r_n(\ell)|=k)=1\ .
\end{align*}
Now, we assess the scaling behavior of $\bar R_n$ which in turn provides a lower bound on the throughput scaling of our proposed scheduling procedure. For this purpose, based on \eqref{eq:Rm2_3}, we need to obtain the distributions of order statistics $\{\sinr^{(1)}_n(r,\ell),\dots,\sinr^{(K)}_n(r,\ell)\}$ for $n\in\{1,\dots,M\}$, $r\in\{1,\dots,Q\}$, and $\ell\in\{1,\dots,N_t\}$. As mentioned earlier the $\sinr$ terms $\{\sinr^{1}_n(r,\ell),\dots,\sinr^{K}_n(r,\ell)\}$ are independent but{\em not} identical and thereof characterizing the distributions of their associated sequence of order statistics does not seem tractable. For tractability purposes, for each $n\in\{1,\dots,M\}$, $r\in\{1,\dots,Q\}$, and $\ell\in\{1,\dots,N_t\}$ we construct two other sets one containing lower bounds and the other containing upper bounds on $\{\sinr^{1}_n(r,\ell),\dots,\sinr^{K}_n(r,\ell)\}$. We define
\begin{align*}
    \zeta_n^{1}=\min_{m,k,r}\gamma^{k,r}_{n,m}\qquad \mbox{and}\qquad\eta_n^{1}&=\min_{k,r}\gamma^{k,r}_{n,n}\ ,\\
    \zeta_n^{2}=\max_{m,k,r}\gamma^{k,r}_{n,m}\qquad \mbox{and}\qquad\eta_n^{2}&=\max_{k,r}\gamma^{k,r}_{n,n}\ .
\end{align*}
Based on these we define the following lower bound on $\sinr^k_n(r,\ell)$
\begin{equation}\label{eq:sinr_L}
    S^k_n(r,\ell)\dff \frac{|\bH^{k,r}_{n,n}\bomega^r_n(\ell)|^2} {\frac{\zeta_n^{2}} {\eta_n^{1}}\left(\sum_{m,r',\ell'}|\bH^{k,r'}_{n,m}\bomega^{r'}_n(\ell')|^2-|\bH^{k,r}_{n,n}\bomega^r_n(\ell)|^2\right)+\frac{1}{\rho\eta_n^{1}}}\ ,
\end{equation}
and the following upper bound on $\sinr^k_n(\ell)$
\begin{equation}\label{eq:sinr_L}
    T^k_n(r,\ell)\dff \frac{|\bH^{k,r}_{n,n}\bomega^r_n(\ell)|^2} {\frac{\zeta_n^{1}} {\eta_n^{2}}\left(\sum_{m,r',\ell'}|\bH^{k,r'}_{n,m}\bomega^{r'}_n(\ell')|^2-|\bH^{k,r}_{n,n}\bomega^r_n(\ell)|^2\right)+\frac{1}{\rho\eta_n^{2}}}\ ,
\end{equation}
Next we show that $S^k_n(r,\ell)$ and $T^k_n(r,\ell)$ and their associated order statistics have tractable distributions.
\begin{lemma}[$\sinr$ Bounds Distributions]
\label{lemma:CDF_bounds}
For any given $n\in\{1,\dots,M\}$, $r\in\{1,\dots,Q\}$, and $\ell\in\{1,\dots,N_t\}$, the elements of the set $\{S^k_n(r,\ell)\}_{k=1}^K$ are i.i.d. with CDF
\begin{equation}\label{eq:S_dist}
S^k_n(r,\ell) \; \sim \; F_{1}(x;n)\dff 1-\frac{e^{-x/\rho\eta^{1}_n}}{\left(\frac{\zeta_n^{2}} {\eta_n^{1}}\; x+1\right)^{MQN_t-1}}\ ,
\end{equation}
and the elements of the set $\{T^k_n(r,\ell)\}_{k=1}^K$ are i.i.d. with CDF
\begin{equation}\label{eq:T_dist}
T^k_n(r,\ell) \; \sim \; F_{2}(x;n)\dff 1-\frac{e^{-x/\rho\eta^{2}_n}}{\left(\frac{\zeta_n^{1}} {\eta_n^{2}}\; x+1\right)^{MQN_t-1}}\ .
\end{equation}
\end{lemma}
\begin{proof}
See Appendix \ref{app:lemma:CDF_bounds}.
\end{proof}
Assessing the scaling of $\bar R_n$  defined in \eqref{eq:Rm2_3}  (which is a lower bound for $\bbe_\bH[R_n]$) requires the knowledge of the distributions of the order statistics $\{\sinr^{(1)}_n(r,\ell),\dots,\sinr^{(K)}_n(r,\ell)\}$. In the next lemma it is shown how $\sinr^{(j)}_n(r,\ell)$ is related to the order statistics of the sets $\{S^k_n(r,\ell)\}_{k=1}^K$ and $\{T^k_n(r,\ell)\}_{k=1}^K$. We denote the $j^{th}$ largest elements of the sets $\{S^k_n(r,\ell)\}_{k=1}^K$ and $\{T^k_n(r,\ell)\}_{k=1}^K$ by $S^{(j)}_n(r,\ell)$ and $T^{(j)}_n(r,\ell)$, respectively.
\begin{lemma}[Ordered $\sinr$ Bounds]
\label{lemma:order}
For any $n\in\{1,\dots,M\},\;r\in\{1,\dots,Q\}$, and $\ell\in\{1,\dots,N_t\}$, we have
\begin{equation*}
    \forall j\in\{1,\dots,K\}:\quad S^{(j)}_n(r,\ell)\leq\sinr^{(j)}_n(r,\ell)\leq
    T^{(j)}_n(r,\ell)\ .
\end{equation*}
\end{lemma}
\begin{proof}
See Appendix \ref{app:lemma:order}.
\end{proof}
By invoking Lemma \ref{lemma:order} (ordered $\sinr$ bounds), from \eqref{eq:Rm2_3} we find that if $\forall k,n$, $\beta(n,k)\geq 1$, then we have following bounds on $\bar R_n$.
\begin{align}
    \label{eq:rate_sum_bound}
    \sum_{r,\;\ell,\;j}\; q^j_n(r,\ell)\; \bbe_\bH\left[\log\left(1+S^{(j)}_n(r,\ell)\right)\right]\leq
    \bar R_n \leq \sum_{r,\;\ell,\;j}\; q^j_n(r,\ell)\; \bbe_\bH\left[\log\left(1+T^{(j)}_n(r,\ell)\right)\right].
\end{align}
By denoting the CDFs of $S^{(j)}_n(r,\ell)$ and $T^{(j)}_n(r,\ell)$ by $F^{(j)}_1(x;n)$ and $F^{(j)}_2(x;n)$, respectively, \eqref{eq:rate_sum_bound} can be restated as
\begin{align}
    \label{eq:rate_sum_bound2}
    \sum_{r,\;\ell} \int_0^{\infty}& \log(1+x)\; d\bigg(\sum_{j=1}^K q^j_n(r,\ell)F^{(j)}_1(x;n)\bigg)\leq
    \bar R_n \leq \sum_{r,\;\ell} \int_0^{\infty}\log(1+x)\;d\bigg(\sum_{j=1}^K q^j_n(r,\ell)F^{(j)}_2(x;n)\bigg) .
\end{align}
Note that
\begin{equation*}
    \{q^j_n(r,\ell)\}_{j}\mbox{  is a PMF}\quad\Rightarrow\quad \sum_{j=1}^K q^j_n(r,\ell)\; F^{(j)}_1(x;n)\;\;\mbox{and}\;\; \sum_{j=1}^K q^j_n(r,\ell)\; F^{(j)}_2(x;n)\;\;\mbox{are CDFs}\ ,
\end{equation*}
and thereof we define the new CDFs
\begin{equation*}
    F^K_{1}(x;n,r,\ell)\dff\sum_{j=1}^K q^j_n(r,\ell)\;F^{(j)}_1(x;n)\quad
    \mbox{and}\quad F^K_{2}(x;n,r,\ell)\dff\sum_{j=1}^K q^j_n(r,\ell)\;F^{(j)}_2(x;n)\ .
\end{equation*}
Hence, we can further reformat \eqref{eq:rate_sum_bound2} as
\begin{align}
    \label{eq:rate_sum_bound3}
    \sum_{r,\;\ell} \int_0^{\infty} \log(1+x)\; &dF^K_{1}(x;n,r,\ell)\;\;\leq \;\;
    \bar R_n \;\; \leq \;\; \sum_{r,\;\ell} \int_0^{\infty}\log(1+x)\;dF^K_{2}(x;n,r,\ell)\ .
\end{align}
In the next step, we define an auxiliary random variable $X$ distributed exponentially with unit mean and denote its CDF by
\begin{equation*}
    G(x)\dff1-e^{-x}\ .
\end{equation*}
Also let $G^{(j)}(x)$ denote the CDF of the $j^{th}$ order
statistic of a set of statistical samples with $K$ members with the  parent distribution $G(x)$. Similar to $F_1^K(x;n,\ell)$ and $F_2^K(x;n,\ell)$ define
\begin{equation*}
    G^K(x;n,r,\ell)\dff\sum_{j=1}^Kq^j_n(r,\ell)\; G^{(j)}(x)\ .
\end{equation*}
The motivation for defining this auxiliary exponential random variable is to ultimately obtain lower and upper bounds on $\bar R_n$ that can be characterized as functions of $G(x)$. Such characterization is instrumental in the sense that we will be able to use some well-known properties of the order statistics of exponential distribution for finding the scaling behavior of $\bar R_n$. The next lemma provides the key step for connecting the bounds given in \eqref{eq:rate_sum_bound3} with new bounds that are functions of $G(x)$.

\begin{lemma}[Bounds on $\bar R_n$]
\label{lemma:bound}
If $\beta^k_n(r)\geq 1$ $\forall n\in\{1,\dots,M\}$, $\forall k\in\{1,\dots,K\}$, and $\forall r\in\{1,\dots,Q\}$, then we have
\begin{align}
    \label{eq:rate_sum_bound4} \sum_{r,\;\ell}&\int_0^{\infty}\log(1+\rho\eta^1_nx)\;dG^K(x;n,r,\ell)-A_n\leq \bar R_n\leq \sum_{r,\;\ell} \int_0^{\infty}\log(1+\rho\eta^2_nx)\;dG^K(x;n,r,\ell)\ ,
\end{align}
where $A_n\dff QN_t\log\Big((MQN_t-1)\rho\zeta^2_n\eta^1_n+1\Big)$.
\end{lemma}
\begin{proof}
See Appendix \ref{app:lemma:bound}.
\end{proof}
To this end we have established lower and upper bounds on $\bar R_n$ based on the order statistics of the exponential distribution. As the final step, we assess the scaling of these bounds. The next lemma reveals how the scaling of the bounds on $\bar R_n$ provided in Lemma \ref{lemma:bound} (bounds on $\bar R_n$), pertains to the characterization of the CDF $G^K(x;\ell)$, which itself depend on the PMFs $\{q^j_n(r,\ell)\}_j$.

\begin{lemma}[Scaling]
\label{lemma:exp_scaling2}
For a sequence of $K$ i.i.d unit-variance exponentially distributed random variables with CDF $W(x)$ and for any arbitrary PMF $\{Q_j\}_{j=1}^K$ if the condition
\begin{equation}\label{eq:Q1}
    \lim_{K\rightarrow\infty}\frac{\sum_{j=1}^KjQ_j}{K}=0\ ,
\end{equation}
is satisfied, then for any positive real number $a\in\mathbb{R}_+$ we have
\begin{equation*}
    \int_0^\infty\log(1+a x)dW^K(x)\doteq\log\log K+\log a\ ,
\end{equation*}
where
\begin{equation*}
    W^K(x)\dff\sum_{j=1}^KQ_jW^{(j)}(x)\ ,
\end{equation*}
and $W^{(j)}$ is the CDF of the $j^{th}$ order statistic of the set of these $K$ random variables.
\end{lemma}
\begin{proof}
See Appendix \ref{app:lemma:exp_scaling2}.
\end{proof}
By using the results of Lemmas \ref{lemma:bound} (bounds on $\bar R_n$) and \ref{lemma:exp_scaling2} (scaling) we state the main result of this paper in the next theorem.

\begin{theorem}[MISO Capacity Scaling]
\label{th:sum_rate}
For downlink transmission in a MISO network with $M$ super-cells each with $Q$ base-stations, $K$ users per super-cell and $N_t$ transmit antenna per base-station, the sum-rate throughput achievable via the proposed scheduling procedure scales as
\begin{equation*}
    \bbe_{\bH}[R]\doteq MQN_t\log\log K\ .
\end{equation*}

\end{theorem}
\begin{proof}
We start by showing that the choice of $Q_j=q^j_n(r,\ell)$ for  $j\in\{1,\dots,K\}$, where $\{q^j_n(r,\ell)\}$ is defined in \eqref{eq:Rm2_1}, fulfils the condition (\ref{eq:Rm2_22}) of Lemma
\ref{lemma:exp_scaling2} (scaling) for any $n\in\{1,\dots,M\},\;r\in\{1,\dots,Q\}$, and $\ell\in\{1,\dots,N_t\}$. From \eqref{eq:Rm2_22} we have
\begin{align}
    \nonumber \sum_{j=1}^Kj Q_j &= \sum_{j=1}^Kj\; q^j_n(r,\ell)= \sum_{j=1}^Kj\sum_{k=j}^K\frac{P(|\CH^r_n(\ell)|=k)}{k}\\
    \nonumber &=
    \sum_{k=1}^K\frac{1}{k}P(|\CH^r_n(\ell)|=k)\sum_{j=1}^{k}j\\
    \nonumber &=
    \sum_{k=1}^K\frac{k+1}{2}P(|\CH^r_n(\ell)|=k)\\
    \nonumber&=
    \frac{1}{2}\underset{\bbe[|\CH^r_n(\ell)|]}{\underbrace{\sum_{k=1}^K\; k\; P(|\CH^r_n(\ell)|=k)}}+
    \frac{1}{2}\underset{=1}{\underbrace{\sum_{k=0}^K P(|\CH^r_n(\ell)|=k)}}-\frac{1}{2}P(|\CH^r_n(\ell)|=0)\\
    \label{eq:Q2} & \leq \frac{1}{2}\;\bbe[|\CH^r_n(\ell)|] +\frac{1}{2}\ .
\end{align}
Note that the probability that each individual user in the $n^{th}$ super-cell expresses interest in the $\ell^{th}$ beam of $B^r_n$, i.e., $\bomega^r_n(\ell)$, by feeding back its index is $\frac{1}{K}$. As shown at the beginning of Appendix~\ref{app:lemma:sets} we have
\begin{equation*}
    \mbox{if}\;\;\beta^k_n(r)\geq 1\quad\Rightarrow\quad \Big\{\quad \sinr^k_n(r,\ell)\geq\beta^k_n(r)\quad\Leftrightarrow\quad k\in\CH^r_n(\ell)\quad\Big\}\ .
\end{equation*}
Therefore,
\begin{equation*}
    P(k\in\CH^r_n(\ell))=P\Big(\sinr^k_n(r,\ell)\geq\beta^k_n(r)\Big) =1-F^k_n(\sinr^k_n(r,\ell);r) \overset{\eqref{eq:lambda}}{=}\frac{1}{K}\ .
\end{equation*}
Since $|\CH^r_n(\ell)|$ counts the number of users in the $n^{th}$ super-cell that have expressed interest in beamformer $\bomega^r_n(\ell)$, we have
\begin{equation}\label{eq:Q3}
    |\CH^r_n(\ell)|\sim\mbox{Binomial}(K,1/K)\quad\Rightarrow\quad \bbe[|\CH^r_n(\ell)|]=K\cdot\frac{1}{K}=1\ .
\end{equation}
Equations \eqref{eq:Q2} and \eqref{eq:Q3} provide that
\begin{equation*}
    \sum_jjQ_j\leq 1\quad\Rightarrow\quad \lim_{K\rightarrow\infty}\frac{\sum_{j=1}^KjQ_j}{K}=0\ .
\end{equation*}
Therefore, the necessary condition of Lemma~\ref{lemma:exp_scaling2} (scaling) given in \eqref{eq:Q1} is satisfied and we can apply Lemma~\ref{lemma:exp_scaling2} (scaling) on the bounds on $\bar R_n$ provided in Lemma~\ref{lemma:bound} (bounds on $\bar R_n$). By setting
\begin{equation*}
    W(x)=G(x)\quad,\forall j\;Q_j=q^j_n(r,\ell),\quad\mbox{and}\quad a=\rho\eta^1_n, \rho\eta^2_n\ ,
\end{equation*}
according to Lemma~\ref{lemma:bound} (bounds on $\bar R_n$) we find
\begin{align*}
    QN_t \log\log K+ QN_t \log\left(\frac{\rho\eta^1_n}{(MQN_t-1)\rho\zeta^2_n\eta^1_n+1}\right) \leq \bar R_n  \leq QN_t  \log\log K + QN_t \log(\rho\eta^2_n)\ .
\end{align*}
The bounds above demonstrate that for fixed $M,N_t,\rho,\eta^1_n$ and $\eta^2_n$, and for increasing $K$ we have
\begin{equation*}
    \lim_{K\rightarrow\infty}\frac{\bar R_n}{QN_t\log\log K}= 1\quad\Rightarrow\quad \bar R_n\doteq QN_t\log\log K\ .
\end{equation*}
Taking into account that the sum-rate throughput cannot scale faster than what is achieved via dirty paper coding, i.e., $MQN_t\log\log K$, the proof is complete.
\end{proof}

\begin{figure}[t]
  \centering
  \includegraphics[width=3.5 in]{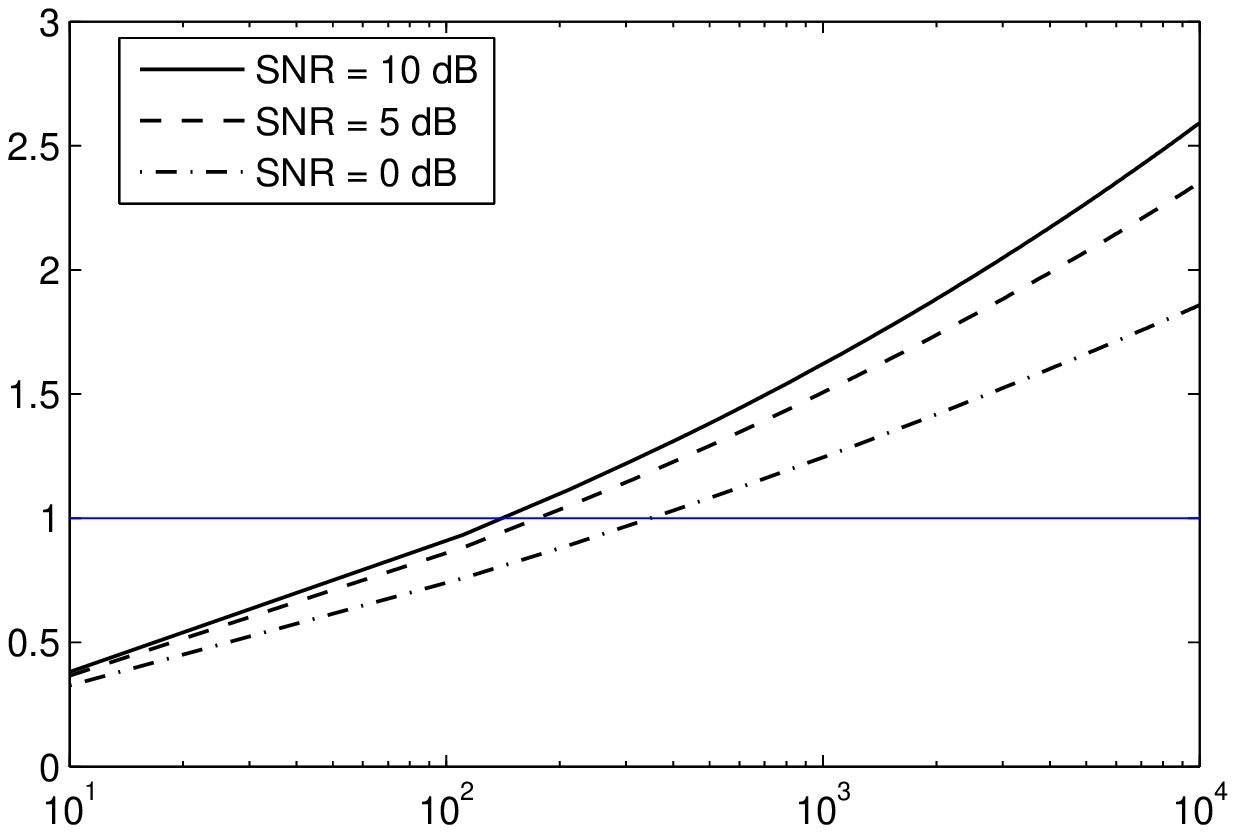}\\
  \caption{Normalizing factors $\beta^k_n(r)$ versus the number of users per super-cell $K$, for $M=3$, $Q=2$ $N_t=2$ and $\gamma^{k,r}_{n,m}=1$.}\label{fig:threshold}
   \begin{picture}(1,1)
      \put(-60,25){{\small\footnotesize Number of users per super-cell $K$}}
      \put(-120,60){\rotatebox{90}{\footnotesize Normalizing factors $\beta^k_n(r)$}}
   \end{picture}
\end{figure}

All the analyses are constructed on the basis that $\forall k,n,r$, we have $\beta^k_n(r)\geq 1$. As discussed earlier, under this condition, each user will have only one of its normalized $\sinr$s (corresponding to different beamformers) greater than 1. Also, as implied throughout the analysis, at the asymptote of large $K$, in each super-cell only $QN_t$ users will participate in the feedback processes, which is the core reason of having controlled amount of feedback. The condition $\beta^k_n(r)\geq 1$ can be guaranteed in the asymptote of large $K$ by recalling that $F^k_n(\beta^k_n(r);r)=1-\frac{1}{K}$. By numerically evaluating $F^k_n(\beta^k_n(r);r)$ through Monte Carlo realizations, Fig.~\ref{fig:threshold} shows how the normalization factor changes with the number of users at $\snr$s. It is seen that in a network with $M=3$ super-cells, each containing $Q=2$ cells and each base-station equipped with $N_t=2$ transmit antennas the normalizing factors $\beta^k_n(r)$ very quickly exceed 1 (as the requirement of Lemma~\ref{lemma:sets} on the conditional expected rates) for different $\snr$ levels $0, 5$ and 10 dB. The normalization factors are obtained via numerically solving \eqref{eq:lambda}.

\subsection{Throughput Scaling for $N_r>1$}
In this section we generalize the result of Theorem~\ref{th:sum_rate} (MISO capacity scaling) to MIMO networks with arbitrary number of receive antennas. As stated earlier in~(\ref{eq:DPC_scale}), by employing dirty-paper coding and facilitating full CSI feedback, the sum-rate throughput exhibits a double logarithmic growth with the number of receive antennas, $N_r$. We show that with slight modification to the scheduling procedure provided in Section~\ref{sec:algorithm} the same gain can be retained.

We modify the scheduling procedure such that it allows each user to receive more than one information stream across its different receive antennas. In other words, we preclude the receive antennas to jointly decode the information streams and consider different receive antennas as separate users. As a result, this translates to having effectively $KN_r$ users in the network and all the analyses provided earlier for $K$ users can be extended for the network with $KN_r$ users. It is noteworthy that although it might not be the most effective way to exploit the receive antennas, it yet achieves the optimal capacity scaling. Indeed, more sophisticated approaches involving joint signal decoding by the antenna arrays are expected to achieve the same gain. Based on this discussion, we modify the normalization factors $\beta^k_n(r)$ such that they satisfy
\begin{equation}
    \label{eq:threshold2}
    F^k_n(\beta^k_n(r);r)=1-\frac{1}{KN_r}\ ,
\end{equation}
and use the same normalization factor for all antennas of the same user. According to the modified scheduling algorithm on one
hand, each user might be required to feed back the indices of more than one beam, which potentially increases the amount of feedback
amount per user. On the other hand, increasing the normalization factor makes it more stringent for the users to satisfy their designated normalization factor constraints and send feedback. The discussions on the aggregate amount of feedback and effects of the number of users and number of receive antennas are given in Section~\ref{sec:feedback}. As a direct result of Theorem~\ref{th:sum_rate} (MISO capacity scaling) we obtain the following scaling law for MIMO networks.
\begin{theorem}[MIMO Capacity Scaling]
\label{th:sum_rateN}
For downlink transmission in a MIMO network with $M$ super-cells each with $Q$ base-stations, $K$ users per super-cell, $N_t$ transmit antenna per base-station, and $N_r$ receive antennas pre user, the sum-rate throughput achievable via the proposed scheduling procedure scales as
\begin{equation*}
    \bbe_{\bH}[R]\doteq MQN_t\log\log KN_r\ .
\end{equation*}
\end{theorem}
We remark that in the scheduling procedure we have considered that each user may receive more than one information stream. On the other hand, if we restrict each user to receive no more than one information stream, the performance  will degrade. Indeed, as shown in~\cite{Sharif:IT05} for single-cell broadcast transmissions, even with perfect $\sinr$ feedback, the sum-rate throughput scales as $N_t\log\log K$ which does not depend on the number of receive antennas. Therefore, with such restriction, adding more receiver antennas offers no asymptotic gain. Simulation results in Fig. \ref{fig:cap} depict the sum-rate throughput achieved by the proposed scheduling scheme and show that it varies linearly with the number of super-cells $M$ and transmit antennas $N_t$ and double logarithmically with the number of users $K$.
\begin{figure}[t]
  \centering
  \includegraphics[width=3.5 in]{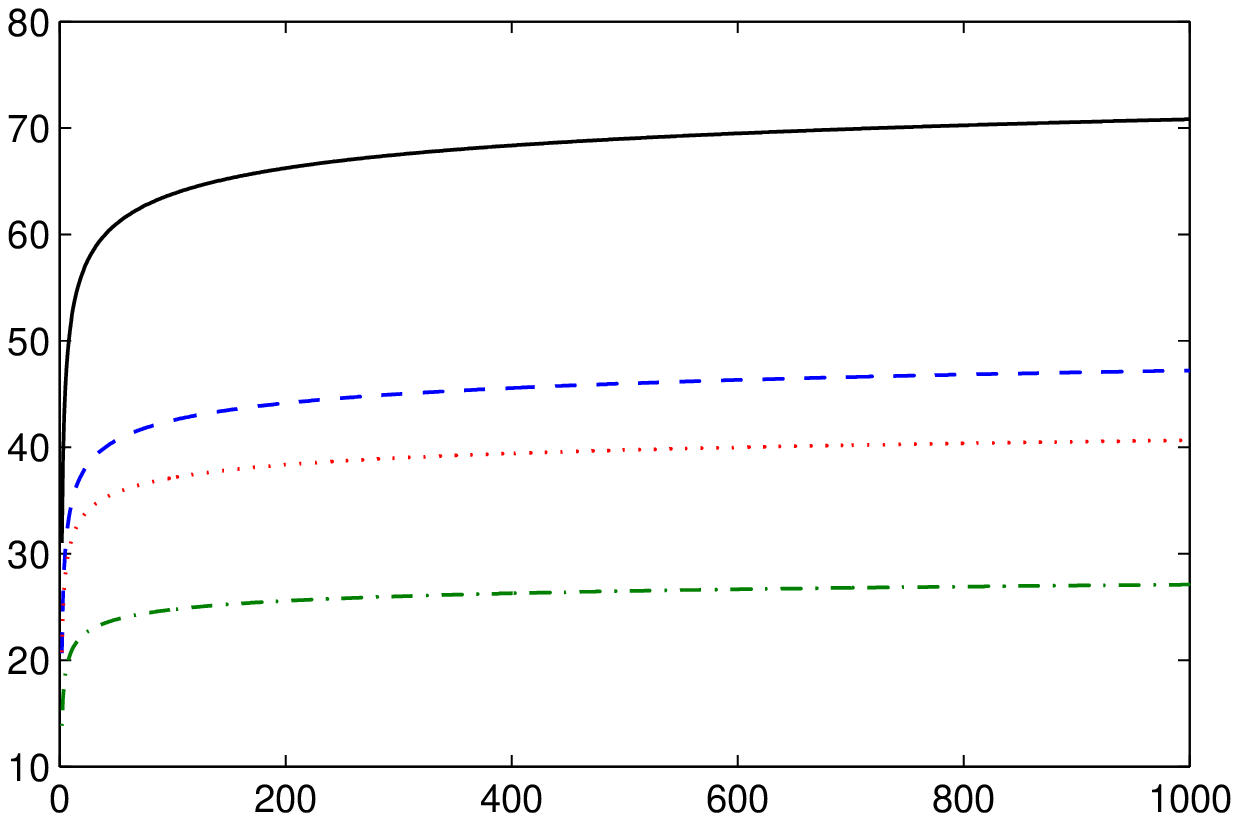}\\
  \caption{Sum-rate throughput versus the number of users per super-cell $K$, for $Q=2$ and $\rho=10$ dB. Each user has $N_r=1$ receive antenna.}\label{fig:cap}
  \begin{picture}(0,0)
      \put(-60,25){{\footnotesize Number of users per super-cell $K$}}
      \put(-120,50){\rotatebox{90}{\footnotesize Sum-rate Throughput (bits/sec/Hz)}}
      \put(0,150){\footnotesize $M=2$, $N_t=3$}
      \put(0,120){\footnotesize $M=2$, $N_t=2$}
      \put(0,95){\footnotesize $M=1$, $N_t=3$}
      \put(0,72){\footnotesize $M=1$, $N_t=2$}
   \end{picture}
\end{figure}

Finally, as a corollary to Theorem~\ref{th:sum_rateN} we can recover the existing results on the capacity scaling of the {\em single}-cell broadcast channels. The result provided in the following corollary indicates that deploying the proposed scheduling scheme for only one super-cell, requires only {\em finite}-rate feedback and yet retains the optimal capacity scaling of the broadcast channels with full channel state information feedback. Similar to the results for the multi-cell networks, in the following corollary we also assume that the number of transmit and receive antennas $N_t,N_r$ are fixed.
\begin{coro}[Single-cell Capacity Scaling]
\label{cor:sum_rateN}
For a MIMO broadcast channel with $N_t$ transmit antenna that serves $K$ users each equipped with $N_r$ receive antennas, the sum-rate throughput achievable via the proposed scheduling procedure scales as
\begin{equation*}
    \bbe_{\bH}[R]\doteq N_t\log\log KN_r\ .
\end{equation*}
\end{coro}

\section{Aggregate Feedback}
\label{sec:feedback}

So far we have established the sum-rate throughput scaling. In this section we evaluate the expected amount of feedback bits imposed by the proposed scheduling procedure. This essentially provides an upper bound on the minimum required amount of information exchange in the network. As discussed in  Section~\ref{sec:algorithm}, if a user meets the criterion to be considered as a candidate for being selected and scheduled, it will feed back $ \lceil\log QN_t\rceil $ bits per eligible receive antenna, and otherwise will remain silent indicating that it is not a candidate for being scheduled. Therefore, each user, potentially, might send as many as $\min(N_t,N_r) \lceil\log QN_t\rceil $ bits via the feedback channel. This, in turn, can potentially give rise to up to $K\min(N_t,N_r) \lceil\log QN_t\rceil $ aggregate feedback bits in each super-cell, which can be prohibitive. Therefore, it is imperative to examine and control the burden put on the feedback channel.

Note that the normalization factors $\beta^k_n(r)$ change adaptively with varying $K$. Therefore, as the number of users per super-cell increases, the chance that each individual user satisfies the normalization constraint becomes slimmer. In other words, the choices of $\beta^k_n(r)$ guarantee that as the network becomes dense, the aggregate feedback load is controlled via enforcing a more stringent constraint on each user to utilize the feedback channel. In practice, obtaining $\beta^k_n(r)$ through solving \eqref{eq:lambda} can be carried out via bi-sectional search. As the CDF is an strictly increasing function in its argument and the problem in \eqref{eq:lambda} has a unique solution, a simple bi-section search, which is computationally efficient, can solve the problem with arbitrary accuracy. In the next theorem we quantify the aggregate feedback and show that the {\em expected} aggregate
feedback is a constant independent of the number of users in the network.
\begin{theorem}[Aggregate Feedback]
\label{th:feedback}
For downlink transmission in a MIMO network with $M$ super-cells each with $Q$ base-stations, $K$ users per super-cell, $N_t$ transmit antenna per base-station, and $N_r$ receive antennas pre user, by deploying the proposed scheduling procedure, the aggregate feedback in the network is $MQN_t \lceil\log QN_t\rceil $.
\end{theorem}
\begin{proof}
Lets define $N=\min(N_t,N_r)$, which is the maximum number of
information streams a user can receive. Also let $\sinr^k_n(r,\ell,i)$ denote the detection $\sinr$ of $i^{th}$ receive antenna of user $U^k_n$ when served along the beamforming vector $\bomega^r_n(\ell)$. The probability that the $i^{th}$ receive antenna of user $U^k_n$ is a candidate to be served by one of the $QN_t$ beamformers of its designated super-cell, is upper bounded by
\begin{equation*}
    P\Big(\max_{r,\;\ell}\frac{\sinr^k_n(r,\ell,i)}{\beta^k_n(r)}\geq 1\Big)\overset{\tiny\mbox{union bound}}{\leq}\sum_{r=1}^Q\sum_{\ell=1}^{N_t}P\Big(\sinr^k_n(r,\ell,i)\geq\beta^k_n(r)\Big) =\frac{QN_t}{KN_r}\dff \mu\ .
\end{equation*}
Hence, the probability that exactly $j\leq N$ users satisfy the normalization constraint is upper bounded by $\mu_j\dff {N_r\choose j}\mu^j(1-\mu)^{N_r-j}$. Therefore, by denoting the number of feedback bits per user by $R_n^{\tiny{ \rm FB}}$ and taking into account that if a user satisfies the normalization constraint, it will feed back $ \lceil\log QN_t\rceil $ bits, we get
\begin{equation*}
    \bbe[R^{\rm FB}]\leq  \lceil\log QN_t\rceil \sum_{j=1}^{N}j\mu_j= \lceil\log QN_t\rceil \sum_{j=1}^{N}j{N_r\choose
    j}\mu^j(1-\mu)^{N_r-j}.
\end{equation*}
Hence, the aggregate amount of feedback in the $n^{th}$ super-cell, denoted by $\bbe[R^{\rm FB}_n]$ in the asymptote of large $K$ is
\begin{eqnarray}
  \nonumber\bbe[R^{\rm FB}_n] &\leq & \lim_{K\rightarrow\infty}K \lceil\log QN_t\rceil \sum_{j=1}^{N}j{N_r\choose
  j}\mu^j(1-\mu)^{N_r-j}\\
  \nonumber&=& \lceil\log QN_t\rceil \sum_{j=1}^{N}j{N_r\choose j}
  \lim_{K\rightarrow\infty}\frac{\Big(\frac{QN_t}{KN_r}\Big)^j\Big(1-\frac{QN_t}{KN_r}\Big)^{N_r-j}}{\frac
  1 K}\\
  \nonumber&=& \lceil\log QN_t\rceil \sum_{j=1}^{N}j{N_r\choose j} \Big(\frac{QN_t}{N_r}\Big)^j\lim_{K\rightarrow\infty}\frac{1}{K^{j-1}}\\
  \label{eq:FB2}&=&QN_t \lceil\log QN_t\rceil \ ,
\end{eqnarray}
which provides that the expected aggregate amount of feedback bits in the network does not exceed $MQN_t \lceil\log QN_t\rceil $. The last equality holds by noting that for $j>1$ the term $\frac{1}{K^{j-1}}\rightarrow 0$ and the only non-zero term corresponds to $j=1$ and is equal to $QN_t \lceil\log QN_t\rceil$.
\end{proof}

The simulation results provided in Fig. \ref{fig:sum_feedback} quantify the aggregate amount of feedback imposed by our scheduling protocol to the feedback channel. As anticipated by the analyses, as the number of users increase, the total amount of feedback bits is $QN_t \lceil\log QN_t\rceil $ which is independent of the number of users as well as the number of receive antennas. This implies that as the number of users increase, only $QN_t$ users will participate in feeding back their indices of their desired beams. This is due to the fact that high normalizing factors for networks with large number of users, only allow the best $QN_t$ users to feed back the indices of their corresponding beams. For the simulations we consider the settings $Q=1, 2$ base-stations per super-cell and assume that each user has $N_r=3$ receive antennas.

\begin{figure}[t]
  \centering
  \includegraphics[width=3.5 in]{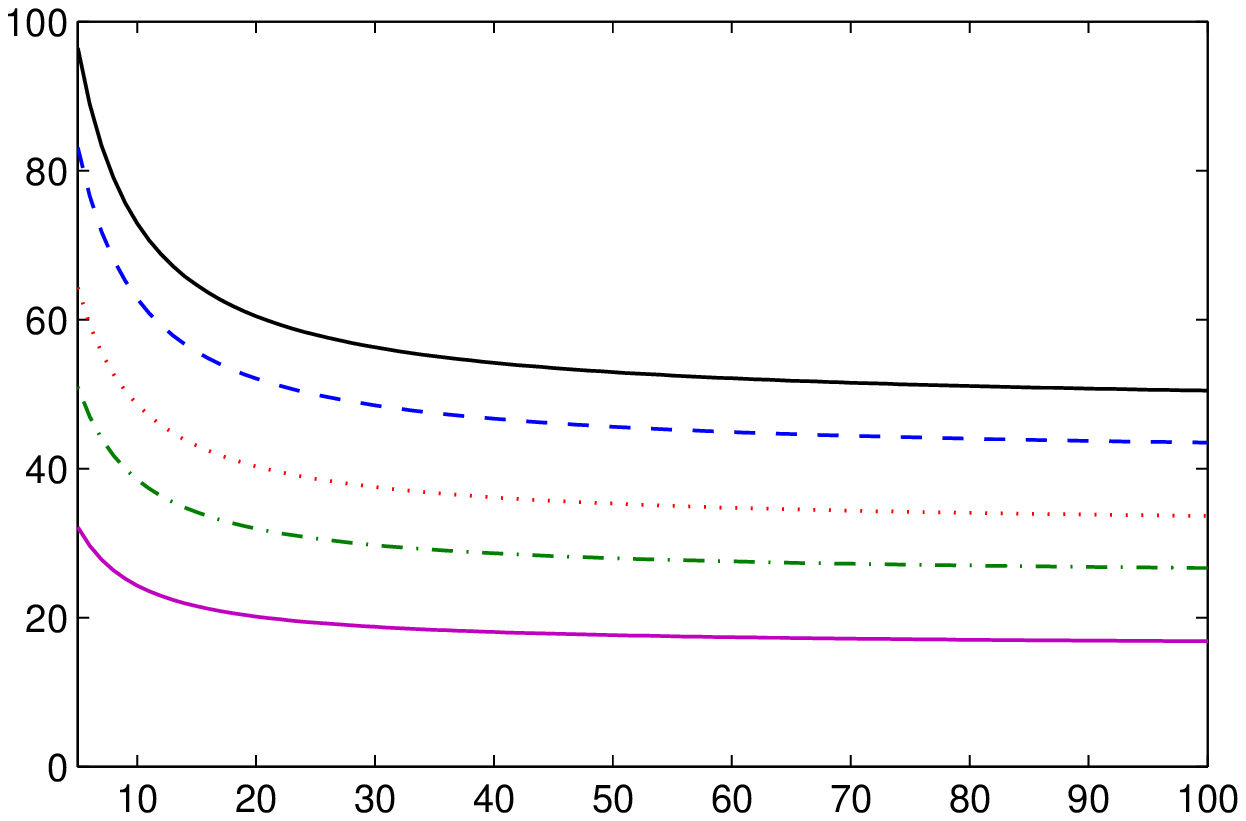}\\
  \caption{Aggregate amount of feedback versus the number of users per super-cell $K$. Each user has $N_r=1$ receive antennas and $\rho=10$ dB.}\label{fig:sum_feedback}
  \begin{picture}(1,1)
      \put(-60,25){{\footnotesize Number of users per super-cell $K$}}
      \put(-120,60){\rotatebox{90}{\footnotesize Sum-rate Feedback $\bbe[R^{\rm FB}_n]$}}
      \put(0,72){\scriptsize $Q=2$, $N_t=4$}
      \put(70,72){\scriptsize $16\log_28$}
      \put(0,85){\scriptsize $Q=2$, $N_t=3$}
      \put(70,85){\scriptsize $12\log_26$}
      \put(0,95){\scriptsize $Q=2$, $N_t=2$}
      \put(70,95){\scriptsize $8\log_24$}
      \put(0,108){\scriptsize $Q=1$, $N_t=3$}
      \put(70,108){\scriptsize $6\log_23$}
      \put(0,120){\scriptsize $Q=1$, $N_t=2$}
      \put(70,120){\scriptsize $4\log_22$}
   \end{picture}
\end{figure}

\section{Practical Issues}
\label{sec:fairness}
\subsection{Size of Super-cells}

The size of the super-cells $Q$ impacts both the scaling rate of the sum-rate throughput (Theorem~\ref{th:sum_rate} on the MISO capacity scaling) as well as the amount of information exchange in the network (Theorem~\ref{th:feedback} on aggregate feedback).
The capacity scaling analysis shows that the sum-rate throughput scales as $MQN_t\log\log KN_r$. This is the scaling of a network with a total number of $MQ$ base-stations and $MK$ users, which means there exist $T\dff \frac{K}{Q}$ users per base-station. When keeping this ratio of users per base-station constant, the sum-rate throughput scales as $MQN_t\log\log QTN_r$. This scaling essentially implies that as long as the product $MQ$ is a constant, the sum-rate scaling remains unchanged. The reason is that for fixed $MQ$, necessarily $M$ will be a finite integer, based on which we have
\begin{equation*}
  \lim_{K\rightarrow\infty}\frac{\bbe_{\bH}[R]}{MQN_t\log\log MKN_r}=\lim_{K\rightarrow\infty}\frac{\log\log KN_r}{\log\log MKN_r}=1,
\end{equation*}
which implies that the throughput of the proposed procedure scales identical to the scaling law provided in Theorem~\ref{th:sum_rateN}. Therefore, when the total number of base-stations in the network $MQ$ is fixed, irrespectively of how we cluster them to construct the super-cells, the scaling law remains fixed. On the other hand, from Theorem~\ref{th:sum_rate} (MISO capacity scaling), the aggregate amount of information exchange is upper bounded by $MQN_t \lceil\log QN_t\rceil $. Once having a fixed number of base-stations in the network, i.e., $MQ$ is fixed, the amount of information exchange is clearly minimized by setting $Q=1$. This choice of $Q$ corresponds to the setting that  there is no cooperation among the base-stations and each base-station acts independently of others. This fact in conjunction with the freedom in clustering the network shows that even when the network is completely distributed (no base-station cooperation) by exchanging as low as $MN_t \log N_t $ information bits, still the optimal sum-rate throughput scaling is achievable.

There exists, however, a cost associated with such distributed downlink transmission. While the sum-rate throughput of the proposes scheme and that of dirty-paper coding scale identically, there exists a gap between them. The distributed processing penalizes the performance by enlarging this gap. By increasing the size of the super-cells $Q$ we can reduce this gap, which shows that we can trade additional information exchange in favor of increasing the achievable sum-rate throughput, which is expected intuitively. The monte-carlo simulation results in Fig.~\ref{fig:cooperation} depicts the impact of clustering on the achievable sum-rate throughput. We consider a network containing $MQ=6$ base-stations and look into the sum-rate of the four possible clustering cases of $(M,Q)=(6,1), (3,2), (2,3)$ and (1,6). While all four cases exhibit identical scalings, it is observed allowing more level of cooperation among the base-stations, i.e., large clusters sizes $Q$, improves the sum-rate throughput. For the simulations we have considered the two cases of $N_t=1,4$ and have assumed that the users have $N_r=3$ receive antennas and $\rho=10$ dB.

\begin{figure}[t]
  \centering
  \includegraphics[width=3.5 in]{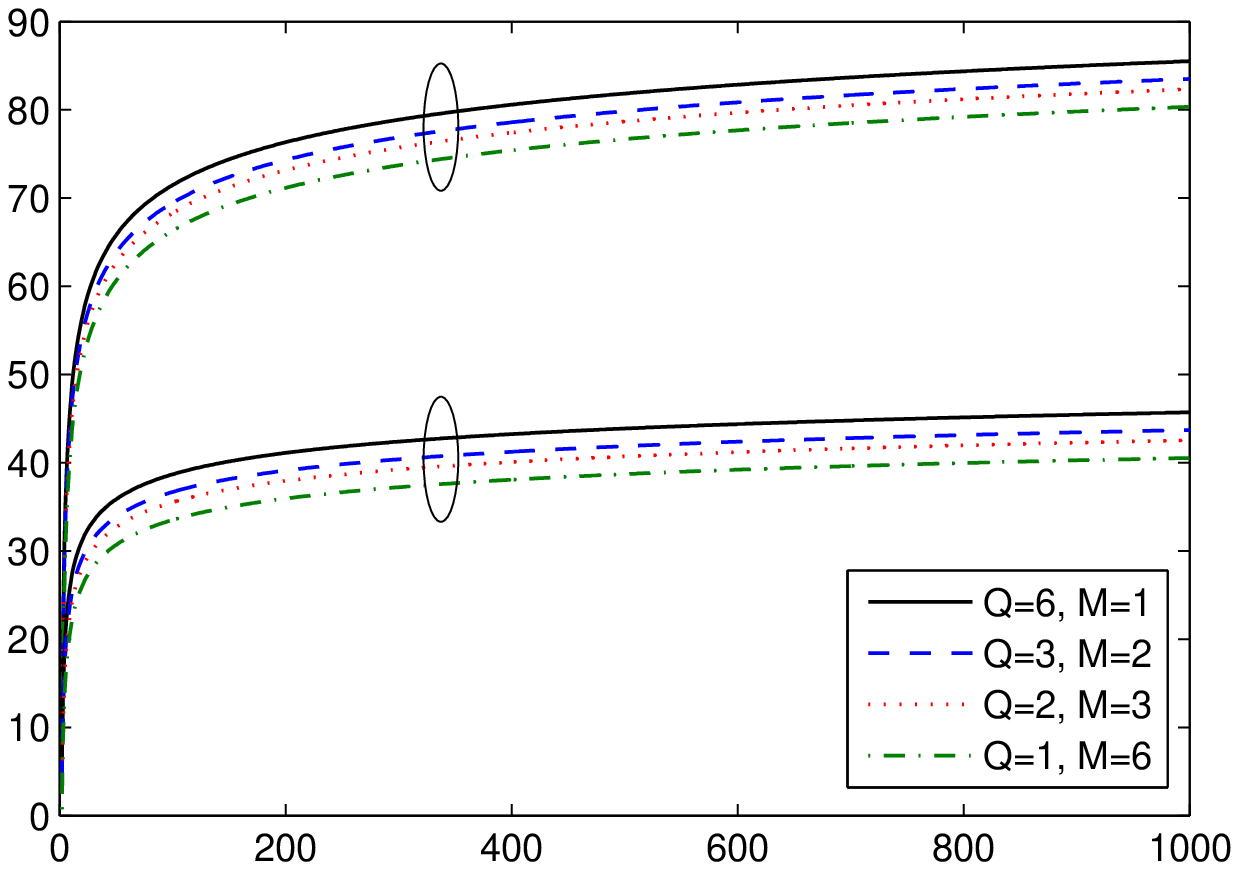}\\
  \caption{The impact of super-cell sizes $Q$ on the achievable sum-rate throughput.}\label{fig:cooperation}
  \begin{picture}(1,1)
      \put(-60,25){{\footnotesize Total number of users in the network $MK$}}
      \put(-120,60){\rotatebox{90}{\footnotesize Sum-rate Throughput (bits/sec/Hz)}}
      \put(-30,90){$N_t=2$}
      \put(-30,150){$N_t=4$}
   \end{picture}
\end{figure}

\subsection{Fairness}

In a non-homogenous network the reception quality at various users not only is influenced by how their channel direction vectors are
aligned with the beam direction, but also is impacted by the path-loss and the shadowing they experience. As scheduling in our scheme involves opportunistic selection of the users who have satisfied a minimum level of reception quality, the system might be dominated by the user with the strongest reception qualities. In general, there exist an inherent tension between opportunistic user selection and maintaining some notion of fairness among the users.  Therefore, it is imperative to investigate how such opportunistic user selection affects fairness.

It is, interestingly, observed that by deploying the proposed scheduling procedure, fairness, in the sense that all users are equiprobable to be served by the base-stations, is maintained. Such fairness is preserved as we assign different normalization factors to the users with different path-loss and shadowing ($\gamma^{k,r}_{m,n}$). This argument is formalized in the next theorem.
\begin{theorem}[Fairness]
\label{th:fairness}

For downlink transmission in a MIMO network with $M$ super-cells each with $Q$ base-stations, $K$ users per super-cell, $N_t$ transmit antenna per base-station, and $N_r$ receive antennas pre user, all users are equally likely to be served by the proposed scheduling procedure.
\end{theorem}
\begin{proof}
The proof relies on the choices of the normalization factors given in \eqref{eq:threshold2}. Let us define $\sinr^k_n(r,\ell,i)$ as received $\sinr$ of the $i^{th}$ receive antenna of user $U^k_n$ when served along the beamforming vector $\bomega^r_n(\ell)$. By following the same line of argument as in \eqref{eq:Q3} and taking into account \eqref{eq:threshold2}, the probability that the $i^{th}$ receive antenna of user $U^k_n$ is a candidate to be scheduled with the beamformer $\bomega^r_n(\ell)$ is
\begin{equation}
    \label{eq:theta}
    \theta_n\dff P\Big(\sinr^k_n(r,\ell,i) \geq \beta^k_n(r)\Big)=1- F^k_n(\beta^k_n(r);r)=\frac{1}{KN_r}\ ,
\end{equation}
which is identical for all users (in the $n^{th}$ super-cell) and all transmit and receive antennas. Therefore, the probability that exactly $j\leq\min(N_t,N_r)$ number of the receive antennas of user $k$ being eligible to be served by a base-station is
\begin{equation*}
    \theta^j_n\dff{N_t\choose j}\theta_n^j(1-\theta_n)^{N_t-j},
\end{equation*}
which is also identical for all users in the $n^{th}$ super-cell. Therefore, all users in each super-cell are equally likely to be scheduled for receiving any certain number of information streams.
\end{proof}

\subsection{Training}
As briefly mentioned at the end of Section~\ref{sec:algorithm} each user needs to calculate the $\sinr$s that each of the beamformers can sustain for carrying the information intended for that user. Acquiring such $\sinr$s is viable during the training interval. Such training being carried out on the downlink channel will cost downlink bandwidth. Nevertheless, it can be readily shown that such bandwidth loss does not penalize the throughput scaling rate. More specifically, due to the broadcast nature of the wireless channels, the base-stations will transmit a common training message for all users. The duration of such training will be {\em fixed} and independent of the number of users in the network. If we denote the fixed throughput loss due to such training by $R_{\sf train}$, then the throughput after penalizing it by the training bandwidth loss will be $\bbe_{\bH}[R]-R_{\sf train}$. $R_{\sf train}$ being fixed provides that for increasing $K$, $\bbe_{\bH}[R]-R_{\sf train}$ will scale similar to $\bbe_{\bH}[R]$. Therefore, while training incurs a loss on the throughput, it will not affect the throughput scaling rate provided in Theorem~\ref{th:sum_rateN}.

\section{Conclusions}
\label{sec:conclusion}

In this paper we have assessed the minimum amount of information exchange required in a multi-cell multiuser downlink transmission with base-station cooperation, in order to retain the optimal capacity scaling. By providing a constructive proof we have demonstrated that by an appropriate design of a feedback mechanism for each super-cell, we can avoid any cross-(super) cellular information exchange and yet achieve the optimal scaling law. More specifically, for a MIMO network consisting of $M$ super-cells, each consisting of $Q$ base-stations that carry out the downlink transmission collaboratively, when each base-station is equipped with $N_t$ transmit antennas, exchanging $MQN_t \lceil\log QN_t\rceil $ information bits across the network suffices to guarantee the achievability of the optimal scaling of the sum-rate throughput. This result relies on a constructive proof, where we offer a scheduling framework that satisfies the information exchange constraints. Furthermore, the scheme also provides good performance for practical networks and sustain fairness among the users in the network.

\appendix

\section{Proof of Lemma~\ref{lemma:sets} (Conditional Expected Rates)}
\label{app:lemma:sets}
{\bf Scheme of the Proof:}\\
\begin{equation}\label{eq:H2}
  \CH^r_n(\ell)=\{k\med \sinrn^k_n(r,\ell)\geq 1\}\ .
\end{equation}
This is due to the fact that when $\beta^k_n(r)\geq 1$, each user in the super-cell $\Q_n$ will have a normalized $\sinr$ greater than 1 for {\em at most} one of $QN_t$ beamformers $\{\bomega^r_n(\ell)\}_{r,\ell}$, i.e., if a user has a normalized $\sinr\geq 1$ for a beamformer, its normalized $\sinr$ for the rest of beamformers is smaller than 1. This can be readily verified by noting that when $\beta^k_n(r)\geq 1$ the we have
\begin{align*}
    \sinrn^k_n(r,\ell)\geq 1\quad &\overset{\eqref{eq:sinrn}}\Rightarrow\quad  \sinr^k_n(r,\ell)\geq \beta^k_n(r)\geq 1\quad\overset{\eqref{eq:sinr}}\Rightarrow \quad \sinr^k_n(r',\ell')<1\quad\forall (r',\ell')\neq(r,\ell)\\
    & \overset{\eqref{eq:sinrn}}\Rightarrow\quad \sinrn^k_n(r',\ell')<1\quad\forall (r',\ell')\neq(r,\ell)
\end{align*}
Given the above characterization of $\CH^r_n(\ell)$ in \eqref{eq:H2}, the rest of the proof has three major steps. In the first step by using the characterization of $\CH^r_n(\ell)$ and some simple manipulation the conditional expected value of the rate $\bbe_\bH[R^r_n(\ell)\med \CH^r_n(\ell)]$ is reformulated. In the second step, we find a lower bound on this conditional expected rate. This lower bound itself is also a conditional expected rate with a different condition compared with $\bbe_\bH[R^r_n(\ell)\med \CH^r_n(\ell)]$. Finally, in the third step we show how removing these new conditions from the conditional expected rates will provide a new lower bound on $\bbe_\bH[R^r_n(\ell)\med \CH^r_n(\ell)]$ which is appropriate related to the order statistics of the set $\{\sinr^1_{n}(r,\ell),\dots,\sinr^K_{n}(r,\ell)\}$.\vspace{.3 in}\\
{\bf Step 1) Reformulating $\bbe_\bH[R^r_n(\ell)\med\CH^r_n(\ell)]$}:\\
By recalling \eqref{eq:Rlm} we have
\begin{eqnarray}\label{eq:Rlm2}
  \bbe_\bH[R^r_n(\ell)\med \CH^r_n(\ell)]= \frac{1}{|\CH^r_n(\ell)|}\sum_{k\in\CH^r_n(\ell)} \bbe_\bH\left[\log\left(1+\sinr^k_n(r,\ell)\right)\;\big |\;\CH^r_n(\ell)\right].
\end{eqnarray}
On the other hand, from \eqref{eq:H2} we have
\begin{equation}\label{eq:H3}
  \forall k\in\CH^r_n(\ell): \sinrn^k_n(r,\ell)\geq 1 \hspace{.2 in}\mbox{and}\hspace{.2 in} \forall k'\notin\CH^r_n(\ell): \sinrn^{k'}_n(\ell) < 1\ .
\end{equation}
By combining \eqref{eq:Rlm2} and \eqref{eq:H3} we find that
\begin{align}
    \nonumber
    &\hspace{-.0 in}\bbe_\bH[R^r_n(\ell)\med \CH^r_n(\ell)]\\
    \nonumber & = \frac{1}{|\CH^r_n(\ell)|}\sum_{k\in\CH^r_n(\ell)}\bbe_{\bH}\bigg[\log\Big(1+\sinr^k_n(r,\ell)\Big)\;\big|\;
    \forall k\in\CH^r_n(\ell): \sinrn^k_n(r,\ell)\geq 1;\\
    \nonumber &  \hspace{3.3 in}\mbox{and}\quad \forall k'\notin\CH^r_n(\ell): \sinrn^{k'}_n(\ell) < 1\;\bigg]\\
    \nonumber &=\frac{1}{|\CH^r_n(\ell)|}\sum_{k\in\CH^r_n(\ell)}\bbe_{\bH}\bigg[\log\Big(1+\sinr^k_n(r,\ell)\Big)\;\big|\;
    \forall k\in\CH^r_n(\ell): \sinr^k_n(r,\ell)\geq \beta^k_n(r);\\
    \label{eq:rate_set2} &\hspace{3.3 in}\mbox{and}\quad \forall k'\notin\CH^r_n(\ell): \sinr^{k'}_n(\ell) < \beta^{k'}_n(r)\bigg]\ ,
\end{align}
where in the first equality we have replaced \eqref{eq:H3} into \eqref{eq:Rlm2} and in the second equality we have invoked that $\sinrn^k_n(r,\ell)=\frac{\sinr^k_n(r,\ell)}{\beta^k_n(r)}$.\vspace{.3 in}\\
{\bf Step 2) Lower Bound on $\bbe_\bH[R^r_n(\ell)\med \CH^r_n(\ell)]$:}\\
In the next step we manipulate the conditions of the expected rate given in \eqref{eq:rate_set2} and obtain a more tractable lower bound on $\bbe_\bH[R^r_n(\ell)\med \CH^r_n(\ell)]$. The core of such manipulation is given in the following lemma.

\begin{lemma}[Conditional Expectation]
\label{lemma:condition}
For a random variable $X$, increasing function $g(\cdot)$
and real values $b\geq a$ that belong to the support of $X$ we have
\begin{equation}\label{eq:lemma:condition}
    \bbe\Big[g(X) \med X\geq b \Big]\geq \bbe\Big[g(X) \med X\geq
    a\Big].
\end{equation}
\end{lemma}
\begin{proof}
Define $Y\dff g(X)$ and denote the CDFs of $X$ and $Y$ by $F_X(x)$ and $F_y(y)$, respectively.
\begin{align}\label{eq:lemma:condition1}
   \bbe\Big[g(X) \med X\geq b \Big] &= \int_yy\;d\big(F_{Y\med X\geq b}(y)\big)=\int_0^1\;F^{-1}_{Y\med X\geq b}(u)\;du\ .
\end{align}
On the other hand, by recalling that $g(X)$ is monotonic in $X$
\begin{align*}
    F_{Y\med X\geq b}(y)&=P(g(X)\leq y\med X\geq b)=\frac{P(b\leq X\leq g^{-1}(y))}{P(b\leq X)}=\frac{F_X(g^{-1}(y))-F_X(b)}{1-F_X(b)}\\
    &\leq\frac{F_X(g^{-1}(y))-F_X(a)}{1-F_X(a)}=P(g(X)\leq y\med X\geq a)\\
    &=F_{Y\med X\geq a}(y)\ ,
\end{align*}
which yields that
\begin{equation}\label{eq:lemma:condition2}
    \forall u\in[0,1]:\quad F^{-1}_{Y\med X\geq b}(u)\geq F^{-1}_{Y\med X\geq a}(u)\ .
\end{equation}
Equations \eqref{eq:lemma:condition1} and \eqref{eq:lemma:condition2} provide that
\begin{equation*}
\bbe\Big[g(X) \med X\geq b \Big] \geq \int_0^1\;F^{-1}_{Y\med X\geq a}(u)\;du= \int_yy\;d\big(F_{Y\med X\geq a}(y)\big)=\bbe\Big[g(X) \med X\geq a \Big]\ .
\end{equation*}
\end{proof}
Next we construct a hypothetical event under which the users in $\CH^r_n(\ell)$ satisfy some other $\sinr$ constraints. For this purpose, we introduce the set of users $\tilde \CH^r_n(\ell)$ corresponding to $\CH^r_n(\ell)$ such that
\begin{equation}\label{eq:tilde_H}
\tilde \CH^r_n(\ell)\dff\{k\med \sinr^k_n(r,\ell)\geq \min_i\beta^i_n(r)\}\ .
\end{equation}
Now consider a hypothetical channel condition for which the event $\tilde \CH^r_n(\ell)$ includes the same set of users as the the event $\CH^r_n(\ell)$ does. Under this hypothetical conditions, i.e., $\tilde \CH^r_n(\ell)=\CH^r_n(\ell)$, by recalling the characterization of $\bbe_\bH[R^r_n(\ell)\med \CH^r_n(\ell)]$ given in \eqref{eq:rate_set2} we set $b=\beta^k_n(r)$ and $a=\min_i\beta^i_n(r)$ and replace the conditions $\sinr^k_n(r,\ell)\geq \beta^k_n(r)$ by $\sinr^k_n(r,\ell)\geq \min_i\beta^i_n(r)$. By noting that $b\geq a$ from Lemma~\ref{lemma:condition} we obtain
\begin{align}
    \nonumber & \bbe_\bH[R^r_n(\ell)\med \CH^r_n(\ell)]\\
    \nonumber &\overset{\eqref{eq:lemma:condition}}{\geq} \frac{1}{|\CH^r_n(\ell)|}\sum_{k\in\tilde \CH^r_n(\ell)}\bbe_{\bH}\bigg[\log\Big(1+\sinr^k_n(r,\ell)\Big)\;\big|\;
    \forall k\in\tilde\CH^r_n(\ell):\sinr^k_n(r,\ell)\geq \min_i\beta^i_n(r); \\
    & \label{eq:rate_set3} \hspace{2.75 in}\mbox{and}\quad \forall k'\notin\tilde\CH^r_n(\ell): \sinr^{k'}_n(\ell) < \beta^{k'}_n(r)\bigg]\ .
\end{align}
Next, by taking into account the statistical independence of $\sinr^k_n(r;\ell)$ and $\sinr^{k'}_n(r;\ell)$ for $k\in \tilde\CH^r_n(\ell)$ and $k'\notin \tilde\CH^r_n(\ell)$, which is due to the statistical independence of channel matrices $\bH^{k,r}_{n,m}$ and $\bH^{k',r}_{n,m}$, changing the conditions enforced on  $\sinr^{k'}_n(r;\ell)$ for $k'\notin \tilde\CH^r_n(\ell)$ does not affect the statistics of $\sinr^{k}_n(\ell)$ for $k\in \tilde\CH^r_n(\ell)$. Therefore, slightly changing the conditions $\sinr^{k'}_n(\ell) < \beta^{k'}_n(r)$ to $\sinr^{k'}_n(\ell) < \min_i\beta^i_n(r)$ for all $k'\notin \tilde\CH^r_n(\ell)$ does not change the distribution of $\sinr^k_n(r,\ell)$ and \eqref{eq:rate_set3} can be equivalently cast as
\begin{align}
    \nonumber & \bbe_\bH[R^r_n(\ell)\med \CH^r_n(\ell)]\\ \nonumber &\geq \frac{1}{|\CH^r_n(\ell)|}\sum_{k\in\tilde\CH^r_n(\ell)}\bbe_{\bH}\bigg[\log\Big(1+\sinr^k_n(r,\ell)\Big)\;\big|\;
    \forall k\in\tilde\CH^r_n(\ell):\sinr^k_n(r,\ell)\geq \min_i\beta^i_n(r);\\
    & \label{eq:rate_set4} \hspace{2.65 in}\mbox{and}\quad \forall k'\notin\tilde \CH^r_n(\ell): \sinr^{k'}_n(\ell) < \min_i\beta^i_n(r)\bigg]\ .
\end{align}
The new condition
\begin{equation*}
    \forall k\in\tilde\CH^r_n(\ell):\sinr^k_n(r,\ell)\geq \min_i\beta^i_n(r)\quad\mbox{and}\quad \forall k'\notin \tilde\CH^r_n(\ell): \sinr^{k'}_n(\ell) < \min_i\beta^i_n(r)\ ,
\end{equation*}
implies that under this condition, the set $\tilde\CH^r_n(\ell)$ contains the $|\tilde\CH^r_n(\ell)|=|\CH^r_n(\ell)|$ largest values of the set $$\{\sinr^1_{n}(r,\ell),\dots,\sinr^K_{n}(r,\ell)\}\ ,$$ and the lower bound on $\bbe_\bH[R^r_n(\ell)\med \CH^r_n(\ell)]$ given in \eqref{eq:rate_set4} is the arithmetic average of these $|\CH^r_n(\ell)|$ largest values conditioned on the constraint that they are larger than or equal to $\min_i\beta^i_n(r)$ and the remaining $(K-|\CH^r_n(\ell)|)$ $\sinr$ values are smaller than $\min_i\beta^i_n(r)$. This interpretation can be formalized as follows.
\begin{align}
    \nonumber & \bbe_\bH[R^r_n(\ell)\med \CH^r_n(\ell)]\\
    \label{eq:rate_set5} & \geq \frac{1}{|\CH^r_n(\ell)|}\sum_{j=1}^{|\CH^r_n(\ell)|}\bbe_{\bH}\bigg[\log\Big(1+\sinr_n^{(j)}(r,\ell)\Big)\;\big|\;
    \sinr_n^{(A)}(r,\ell)\geq\min_i\beta^i_n(r)>\sinr_n^{(A+1)}(r,\ell)\bigg]\ ,
\end{align}
where we have defined $A\dff |\CH^r_n(\ell)|$.\vspace{.3 in}\\
{\bf Step 3) Removing the Condition in \eqref{eq:rate_set5}:}\\
In the final step, we further simplify the lower bound on $\bbe_\bH[R^r_n(\ell)\med \CH^r_n(\ell)]$ provided in \eqref{eq:rate_set5}. This step will be the immediate application of the following lemma.
\begin{lemma}[Conditional Expected Value of the Order Statistics]\label{lemma:order2}
Consider the set of independent random variables $\{X_1,\dots,X_K\}$ and its corresponding set of order statistics $\{X^{(1)},\dots,X^{(K)}\}$. For any increasing function $g(\cdot)$ and the real value $a$ belonging to the supports of $X_k$ we have
\begin{equation}\label{eq:lemma:order1}
    \forall\;  j\leq l:\quad \bbe\left[g\big(X^{(j)}\big) \;\big|\; X^{(l+1)}\leq a\leq X^{(l)}\right]\geq \bbe\left[g\big(X^{(j)}\big)\right]\ .
\end{equation}
\end{lemma}
\begin{proof}
In the first step, by following the same line of argument as in the proof of Lemma~\ref{lemma:condition} (conditional expectation) it can be readily shown that
\begin{equation}\label{eq:lemma:order1_2}
    \forall\; j\leq l:\quad \bbe\left[g\big(X^{(j)}\big) \;\big|\; X^{(l+1)}\leq a\leq X^{(l)}\right]\geq \bbe\left[g\big(X^{(j)}\big) \;\big|\; X^{(j+1)}\leq a\leq X^{(j)}\right]\ .
\end{equation}
In the second step we also prove that
\begin{equation}\label{eq:lemma:order1_3}
    \forall\; j:\quad \bbe\left[g\big(X^{(j)}\big) \;\big|\; X^{(j+1)}\leq a\leq X^{(j)}\right]\geq \bbe\left[g\big(X^{(j)}\big)\right]\ .
\end{equation}
Equations \eqref{eq:lemma:order1_2} and \eqref{eq:lemma:order1_3} together establish the desired result. For proving \eqref{eq:lemma:order1_3} we start by showing that for any real $b\geq a$
\begin{equation}\label{eq:lemma:order2}
    \forall j:\quad P\left(X^{(j)}\leq b \;\big|\; X^{(j+1)}\leq a\leq X^{(j)}\right)\leq P\left(X^{(j)}\leq b\right)\ .
\end{equation}
For this purpose let us define $\pi_i\dff(X_{i_1},\dots,X_{i_K})$ as the $i^{th}$ possible ordered permutations of the set $\{X_1,\dots,X_K\}$ for $i=1,\dots,K!$. By averaging over all such permutations we have $\forall j$
\begin{align}\label{eq:lemma:order3}
     \nonumber &P\left(X^{(j)}\leq b \;\big|\; X^{(j+1)}\leq a\leq X^{(j)}\right)\\
     \nonumber &\hspace{.5 in}=\sum_iP\left(X^{(j)}\leq b \;\big|\; X^{(j+1)}\leq a\leq X^{(j)},\;\;(X^{(1)},\dots,X^{(K)})=\pi_i \right)P\left((X^{(1)},\dots,X^{(K)})=\pi_i\right)\\
     \nonumber &\hspace{.5 in}=\sum_iP\left(X_{i_1},\dots,X_{i_{j-1}}\leq +\infty\;\;\mbox{and}\;\; \;X_{i_j},\dots,X_{i_K}\leq b \;\big|\; X_{i_{j+1}},\dots,X_{i_{K}}\leq a\leq X_{i_1},\dots,X_{i_{j}}\right)\\
     \nonumber&\hspace{1 in}\times P\left((X^{(1)},\dots,X^{(K)})=\pi_i\right)\\
     \nonumber &\hspace{.5 in}=\sum_i\frac{P\left(\{a\leq X_{i_1},\dots,X_{i_{j-1}}\leq +\infty\} \;\;\mbox{and}\;\; \{a\leq X_{i_j}\leq b\}\;\;\mbox{and}\;\; \{X_{i_{j+1}},\dots,X_{i_K}\leq b\}\right)}{P\left(\{a\leq X_{i_1},\dots,X_{i_{j-1}}\leq +\infty\} \;\;\mbox{and}\;\; \{a\leq X_{i_j}\}\;\;\mbox{and}\;\; \{X_{i_{j+1}},\dots,X_{i_K}\leq b\}\right)}\\
     \nonumber&\hspace{1 in}\times P\left((X^{(1)},\dots,X^{(K)})=\pi_i\right)\\
     \nonumber &\hspace{.5 in}=\sum_i\frac{P\left(a\leq X_{i_j}\leq b\right)}{P\left(a\leq X_{i_j}\right)}\times P\left((X^{(1)},\dots,X^{(K)})=\pi_i\right)\\
     \nonumber &\hspace{.5 in}=\sum_i\frac{P\left(X_{i_j}\leq b\right)-P\left(X_{i_j}\leq a\right)}{P\left(a\leq X_{i_j}\right)}\times P\left((X^{(1)},\dots,X^{(K)})=\pi_i\right)\\
     \nonumber &\hspace{.5 in}\leq \sum_i\frac{P\left(X_{i_j}\leq b\right)-P\left(X_{i_j}\leq a\right)P\left(X_{i_j}\leq b\right)}{P\left(a\leq X_{i_j}\right)}\times P\left((X^{(1)},\dots,X^{(K)})=\pi_i\right)\\
     \nonumber &\hspace{.5 in}= \sum_iP\left(X_{i_j}\leq b\right)\times P\left((X^{(1)},\dots,X^{(K)})=\pi_i\right)\\
     \nonumber &\hspace{.5 in}= P\left(X^{(j)}\leq b\right)\ ,
\end{align}
which establishes the claim in \eqref{eq:lemma:order2}. By defining $Y^{(j)}\dff g(X^{(j)})$ and denoting the CDFs of $X^{(j)}$ and $Y^{(j)}$ by $F_{X^{(j)}}(x)$ and $F_{Y^{(j)}}(y)$, respectively, from \eqref{eq:lemma:order2} we immediately have
\begin{equation}\label{eq:lemma:order4}
    F_{Y^{(j)}\med X^{(j+1)}\leq a\leq X^{(j)}}(y)\leq F_{Y^{(j)}}(y)\quad\Rightarrow\quad \forall u\in[0,1]:\quad F^{-1}_{Y^{(j)}\med X^{(j+1)}\leq a\leq X^{(j)}}(u)\geq F^{-1}_{Y^{(j)}}(u)\ .
\end{equation}
Therefore, we get
\begin{align*}
    \bbe\left[g\big(X^{(j)}\big) \;\big|\; X^{(j+1)}\leq a\leq X^{(j)}\right]& =\int_y y\; d\left(F_{Y^{(j)}\med X^{(j+1)}\leq a\leq X^{(j)}}(y)\right)\\
    &=\int_{u=0}^1F^{-1}_{Y^{(j)}\med X^{(j+1)}\leq a\leq X^{(j)}}(u)\;du\\
    &\overset{\eqref{eq:lemma:order4}}{\geq} \int_{u=0}^1F^{-1}_{Y^{(j)}}(u)\;du=\int_y y\; d\left(F_{Y^{(j)}}(y)\right)= \bbe\left[g\big(X^{(j)}\big)\right] \ ,
\end{align*}
that justifies \eqref{eq:lemma:order1_3}, which in turn along with \eqref{eq:lemma:order1_2} concludes the desired result.
\end{proof}
Now, by using the results of the lemmas above and by setting $X^{(j)}=\sinr_n^{(j)}(r,\ell)$ for $j\in\{1,\dots,K\}$, $g(x)=\log(1+x)$, and $l=A=|\CH^r_n(\ell)|$, through invoking \eqref{eq:lemma:order1} (main result of Lemma~\ref{lemma:order2}) we obtain
\begin{align}\label{eq:lemma:order5}
    \nonumber \forall j\leq A:\quad \bbe_{\bH}\bigg[&\log\Big(1+\sinr_n^{(j)}(r,\ell)\Big)\;\big|\;
    \sinr_n^{(A)}(r,\ell)\geq\min_i\beta^i_n(r)>\sinr_n^{(A+1)}(r,\ell)\bigg] \\ & \geq \bbe_{\bH}\bigg[\log\Big(1+\sinr_n^{(j)}(r,\ell)\Big)\bigg]\ .
\end{align}
Combining \eqref{eq:rate_set5} and \eqref{eq:lemma:order5} provides that
\begin{align}\label{eq:rate_set6}
    \bbe_\bH[R^r_n(\ell)\med \CH^r_n(\ell)]
    \geq
    \frac{1}{|\CH^r_n(\ell)|}\sum_{j=1}^{|\CH^r_n(\ell)|}\bbe_{\bH}\bigg[\log\Big(1+\sinr_n^{(j)}(r,\ell)\Big)\bigg]\ ,
\end{align}
which is the desired result.

\section{Proof of Lemma \ref{lemma:CDF_bounds} ($\sinr$ Bounds Distributions)}
\label{app:lemma:CDF_bounds}
Based on these we define the following lower bound on $\sinr^k_n(r,\ell)$
\begin{equation}\label{eq:sinr_L}
    S^k_n(r,\ell)\dff \frac{|\bH^k_{n,n}\bomega^\ell_n|^2} {\frac{\zeta_n^2} {\eta_n^1}\left(\sum_m\sum_l|\bH^k_{n,m}\bomega^l_m|^2-|\bH^k_{n,n}\bomega^\ell_n|^2\right)+\frac{1}{\rho\eta_n^1}}\ ,
\end{equation}
and the following upper bound on $\sinr^k_n(r,\ell)$
\begin{equation}\label{eq:sinr_L}
    T^k_n(r,\ell)\dff \frac{|\bH^k_{n,n}\bomega^\ell_n|^2} {\frac{\zeta_n^1} {\eta_n^2}\left(\sum_m\sum_l|\bH^k_{n,m}\bomega^l_m|^2-|\bH^k_{n,n}\bomega^\ell_n|^2\right)+\frac{1}{\rho\eta_n^2}}\ .
\end{equation}
Let $Y\dff|\bH^{k,r}_{n,n}\bomega^r_n(\ell)|^2$. Since $\bH^{k,r}_{n,n}\bomega^r_n(\ell)\sim{\cal CN}(0,1)$, $Y$ has exponential distribution with unit variance. Also define  $Z\dff\sum_m\sum_{r'}\sum_{\ell'}|\bH^{k,r}_{n,n}\bomega^r_n(\ell)|^2-|\bH^k_{n,n}\bomega^\ell_n|^2$ which is the summation of $MQN_t-1$ independent exponentially distributed random variables each with unit variance. Therefore, $Z$ has a ${\rm Gamma}(MQN_t-1,1)$ distribution. By denoting the probability density functions (PDF) of $Z$ and $Y$ by
\begin{eqnarray*}
  f_Y(y) &=& e^{-z}\quad \mbox{and}\quad f_Z(z) = \frac{z^{MQN_t-2}\;e^{-z}}{(MQN_t-2)!}\ ,
\end{eqnarray*}
the PDF of $T^k_n(r,\ell)=\frac{Y}{1/\rho\eta_n^2+\frac{\zeta_n^1} {\eta_n^2}Z}$, denoted by $f_T(t)$ is
\begin{align*}
  f_T(t) &= \int_0^{\infty}f_{T\med Z}(t\med z)f_Z(z)dz = \int_0^{\infty}\left(\frac{1}{\rho\eta_n^2}+\frac{\zeta_n^1} {\eta_n^2}\; z\right)e^{\left(-\frac{t}{\rho\eta_n^2}+\frac{\zeta_n^1} {\eta_n^2}\; tz\right)}\cdot \frac{z^{MQN_t-2}\;e^{-z}}{(MQN_t-2)!}\;dz\\
  &=\frac{e^{-t/\rho\eta_n^2}}{\left(\frac{\zeta_n^1} {\eta_n^2}t+1\right)^{MQN_t}}
 \left[\frac{\frac{\zeta_n^1} {\eta_n^2}t+1}{\rho\eta_n^2}+\frac{\zeta_n^1} {\eta_n^2}(MQN_t-1)\right],
\end{align*}
where the last step holds as $\int_0^\infty e^{-u}u^M=M!$. Some simple manipulations the CDF of $T^k_n(r,\ell)$ is obtained as
\begin{equation*}
    F_2(x;n)=\int_0^xf_T(t)\;dt= 1-\frac{e^{-x/\rho\eta^2_n}}{(\frac{\zeta_n^1} {\eta_n^2}\; x+1)^{MQN_t-1}}\ .
\end{equation*}
$F_1(x;n)$ can be found by following the same line of argument.

\section{Proof of Lemma~\ref{lemma:order} (Ordered $\sinr$ Bounds)}
\label{app:lemma:order}
By induction we show that for given $n,r,\ell$, for all $j\in\{1,\dots, K\}$, we have $S^{(j)}_n(r,\ell)\leq \sinr^{(j)}_n(r,\ell)$.\\
{\bf 1)} For $j=1$ we have
\begin{equation*}
    \sinr^{(1)}_n(r,\ell)=\max_k\sinr^{k}_n(r,\ell)\geq\max_kS^{k}_n(r,\ell)=S^{(1)}_n(r,\ell).
\end{equation*}
{\bf 2)} Assumption: For some $j=l$ we have
$S^{(l)}_n(r,\ell)\leq \sinr^{(l)}_n(r,\ell)$.\\
{\bf 3)} Claim: For $j=l+1$ we show that
$S^{(l+1)}_n(r,\ell)\leq \sinr^{(l+1)}_n(r,\ell)$.\\
Based on the definition of $S^k_n(r,\ell)$, each of the $(K-l)$ terms $\{\sinr^{(l+1)}_n(r,\ell), \dots, \sinr^{(K)}_n(r,\ell)\}$ is
greater than one corresponding element in the set
$\{S^1_n(r,\ell),\dots,S^K_n(r,\ell)\}$. Therefore, there cannot be more than $l$ elements in the set $\{S^1_n(r,\ell),\dots,S^K_n(r,\ell)\}$ that are all greater than all the members of
$\{\sinr^{(l+1)}_n(r,\ell), \dots, \sinr^{(K)}_n(r,\ell)\}$. Now, by contradiction assume that $S^{(l+1)}_n(r,\ell)> \sinr^{(l+1)}_n(r,\ell)$. This assumption provides that
\begin{equation*}
    S^{(1)}_n(r,\ell)\geq S^{(2)}_n(r,\ell)\geq\dots\geq S^{(l+1)}_n(r,\ell)> \sinr^{(l+1)}_n(r,\ell)\geq\sinr^{(l+2)}_n(r,\ell)\geq \sinr^{(K)}_n(r,\ell)\ ,
\end{equation*}
which indicates that we have found $l+1$ members of the set $\{S^1_n(r,\ell),\dots,S^K_n(r,\ell)\}$ that are greater than all the $(K-l)$ members of the set $\{\sinr^{(l+1)}_n(r,\ell), \dots, \sinr^{(K)}_n(r,\ell)\}$ which contradicts with what we found earlier. Hence, we should have $S^{(l+1)}_n(r,\ell)\leq \sinr^{(l+1)}_n(r,\ell)$. The proof for the inequality $T^{(j)}_n(r,\ell)\geq \sinr^{(j)}_n(r,\ell)$ follows the same line of argument.

\section{Proof of Lemma \ref{lemma:bound} (Bounds on $\bar R_n$)}
\label{app:lemma:bound}
The proof heavily depends on the result of the following lemma.
\begin{lemma}
\label{lemma:f}
For a real variable $x\in[0,1]$ and integer variables $K$ and $j$,
$0\leq j\leq K-1$, the function
\begin{equation*}
    f(x,j)\dff\sum_{i=0}^j{K\choose i}x^{K-i}(1-x)^i
\end{equation*}
is increasing in $x$.
\end{lemma}
\begin{proof}
See Appendix \ref{app:lemma:f}.
\end{proof}

Note that from the properties of the order statistics \cite{Arnold:book} we know that
\begin{align}
  \label{eq:CDF:Lj} F^{(j)}_1(x;n) &\dff \sum_{i=0}^{j-1}{K\choose
  i}\Big(F_1(x;n)\Big)^{K-i}\Big(1-F_1(x;n)\Big)^i\ ,\\
  \label{eq:CDF:Uj} F^{(j)}_2(x;n) &\dff \sum_{i=0}^{j-1}{K\choose
  i}\Big(F_2(x;n)\Big)^{K-i}\Big(1-F_2(x;n)\Big)^i\\
  \label{eq:CDF_G}  G^{(j)}(x) &\dff  \sum_{i=0}^{j-1}{K\choose
    i}\Big(G(x)\Big)^{K-i}\Big(1-G(x)\Big)^i\ .
\end{align}
Also by referring to the definitions of $G(x)$ and $F_1(x;n)$ we find that
\begin{eqnarray*}
  G\bigg((MQN_t-1)\Big(\frac{\zeta^2_n}{\eta^1_n}(1+x)\Big)+\frac{x}{\rho\eta^1_n}\bigg) &=&
  1-\exp\bigg[-\frac{x}{\rho\eta^1_n}-(MQN_t-1)\underset{\geq \ln\Big(1+\frac{\zeta^2_n}{\eta^1_n}x\Big)}{\underbrace{\Big(\frac{\zeta^2_n}{\eta^1_n}(1+x)\Big)}}\bigg] \\
  &\geq&1-\frac{e^{-x/\rho\eta^1_n}}{\Big(1+\frac{\zeta^2_n}{\eta^1_n}x\Big)^{MQN_t-1}}=
  F_1(x;n)\ .
\end{eqnarray*}
By applying Lemma~\ref{lemma:f}, for $j=1,\dots, K$, we have
\begin{eqnarray*}
   G^{(j)}\bigg((MQN_t-1)\Big(\frac{\zeta^2_n}{\eta^1_n}(1+x)\Big)+\frac{x}{\rho\eta^1_n}\bigg)&=&
   f\bigg(G\Big((MQN_t-1)\Big(\frac{\zeta^2_n}{\eta^1_n}(1+x)\Big)+\frac{x}{\rho\eta^1_n}\Big),j-1\bigg)\\
   &\geq& f\Big(F_1(x;n),j-1\Big)\\
   &=& F^{(j)}_1(x;n)\ ,
\end{eqnarray*}
and consequently,
\begin{eqnarray}
    \nonumber G^{K}\bigg((MQN_t-1)\Big(\frac{\zeta^2_n}{\eta^1_n}(1+x)\Big)+\frac{x}{\rho\eta^1_n};n,r,\ell\bigg)&=&
    \label{eq:reverse}\sum_{j=1}^Kq^j_n(r,\ell)G^{(j)}\bigg((MQN_t-1)\Big(\frac{\zeta^2_n}{\eta^1_n}(1+x)\Big)+ \frac{x}{\rho\eta^1_n}\bigg)\\
    &\geq&\sum_{j=1}^Kq^j_n(r,\ell) F^{(j)}_1(x;n)=F^{K}_1(x;n,r,\ell).
\end{eqnarray}
Define $u\dff F^{K}_1(x;n,r,\ell)$. Since $F^{(j)}_1(x;n)$, $j=1,
\dots, K$, are increasing functions in $x$, the function $F^{K}_1(x;n,r,\ell)$ is also increasing in $x$ and thereof convertible. Therefore, (\ref{eq:reverse}) can be rewritten as $x=(F^K_1)^{-1}(u;n,r,\ell)$.
By noting $G^K(x)$ is also invertible we have
\begin{equation*}
     G^{K}\bigg((MQN_t-1)\Big(\frac{\zeta^2_n}{\eta^1_n}(1+(F^K_1)^{-1}(u;n,r,\ell))\Big) +\frac{(F^K_1)^{-1}(u;n,r,\ell)}{\rho\eta^1_n};\ell\bigg)\geq
     u\ ,
\end{equation*}
or equivalently
\begin{equation}
    \label{eq:compare}
    \bigg(\rho\zeta^2_n(MQN_t-1)+1\bigg)\bigg(1+(F^K_1)^{-1}(u;n,\ell)\bigg)\geq
    1+\rho\eta^1_n(G^K)^{-1}(u;n,r,\ell)\ .
\end{equation}
Using (\ref{eq:compare}) we find
\begin{align}
  \nonumber \sum_{r,\;\ell}\int_0^{\infty}& \log(1+x)\;dF^K_1(x;n,r,\ell) + QN_t\log \Big(\rho\zeta^2_n(MQN_t-1)+1\Big)\\
  \nonumber &=\sum_{r,\;\ell}\int_0^1\log\bigg(1+(F^K_1)^{-1}(u;n,r,\ell)\bigg)du+QN_t\log \Big(\rho\zeta^2_n(MQN_t-1)+1\Big)\int_0^1du\\
  \nonumber &=\sum_{r,\;\ell}\int_0^1\log\left[\bigg(\rho\zeta^2_n(MQN_t-1)+1\bigg)\bigg(1+(F^K_1)^{-1}(u;n,r,\ell)\bigg)\right]du\\
  \nonumber &\geq \sum_{r,\;\ell}\int_0^1\log\Big(1+\rho\eta^1_n(G^K)^{-1}(u;n,r,\ell)\Big)du\\
  \label{eq:low_scale} &=\sum_{r,\;\ell}\int_0^{\infty}\log(1+\rho\eta^1_nx)\;dG^K(x;n,r,\ell)\ .
\end{align}
By substituting (\ref{eq:low_scale}) into (\ref{eq:rate_sum_bound3})
we get the following lower bound on $\bar R_n$
\begin{equation}
    \label{eq:low_scale1}
    \sum_{r,\;\ell}\int_0^{\infty}\log(1+\rho\gamma_1x)\;dG^K(x;n,r,\ell)-QN_t\log \Big(\rho\zeta^2_n(MQN_t-1)+1\Big)
    \leq \bar R_n.
\end{equation}
Also it can be easily verified that $\forall x\in\mathbb{R}$,
$F_2(x;n)\geq G\Big(\frac{x}{\rho\eta^2_n}\Big)$. Hence,
\begin{align*}
    F^K_2(x;n,r,\ell)&=\sum_{j=1}^Kq_j^n(r,\ell)\; F^{(j)}_2(x;n)\\
    &=\sum_{j=1}^Kq_j^n(r,\ell)\; f\Big(F_2(x;n),j-1\Big)
    \\
    &\geq \sum_{j=1}^Kq_j^n(r,\ell) \;f\bigg(G\Big(\frac{x}{\rho\eta^2_n}\Big),j-1\bigg)\\
    &= \sum_{j=1}^Kq_j^n(r,\ell)\; G^{(j)}\Big(\frac{x}{\rho\eta^2_n}\Big)=G^K\Big(\frac{x}{\rho\eta^2_n};n,r,\ell\Big),
\end{align*}
By defining $u\dff \frac{F^K_2(x;n,\ell)}{\rho\eta^2_n}$ and
following the same lines as above we ge
\begin{equation}
    \label{eq:high_scale}
    \sum_{r,\;\ell}\int_0^{\infty}\log(1+x)\;dF^K_2(x;n,r,\ell)\leq
    \sum_{r,\;\ell}\int_0^{\infty}\log(1+\rho\gamma_2x)\;dG^K(x;n,r,\ell)\ .
\end{equation}
Inequalities in (\ref{eq:low_scale}) and (\ref{eq:high_scale})
together give rise to the following upper bound on $\bar R_n$
\begin{equation}
    \label{eq:high_scale1}
    \bar R_n\leq
    \sum_{r,\;\ell}\int_0^{\infty}\log(1+\rho\gamma_2x)\;dG^K(x;n,r,\ell)\ .
\end{equation}
Combining the lower and upper bounds given in
(\ref{eq:low_scale1})~and~(\ref{eq:high_scale1}) establishes the
desired result.

\section{Proof of Lemma \ref{lemma:exp_scaling2} (Scaling)}
\label{app:lemma:exp_scaling2}
We start by citing the following theorem.
\begin{theorem}\label{th:limit}
\emph{\cite[Theorem 4]{Sanayei:WCOM07}} Let $\{X_n\}_{n=1}^K$ be a family of positive random variables with finite mean $\mu_K$ and variance $\sigma^2_K$, also $\mu_K\rightarrow\infty$ and
$\frac{\sigma_K}{\mu_K}\rightarrow 0$ as $K\rightarrow\infty$. Then, for all $\alpha\in\mathbb{R}^+$ we have
\begin{equation*}
    \bbe\Big[\log(1+\alpha X_K)\Big]\doteq
    \log\Big(1+\alpha\bbe[X_K]\Big).
\end{equation*}
\end{theorem}
Now, for any given number of users $K$ per super-cell, we define a random variable $X_K$ that is distributed as $X_K\sim W^K(x)$. Also for $j=1,\dots, K$ we define
\begin{eqnarray*}
  \mu_{(j)} &\dff& \int_0^{\infty}x\;dW^{(j)}(x)\ ,\\
  \mbox{and}\;\;\;\sigma^2_{(j)} &\dff&
  \int_0^{\infty}\Big(x-\mu_{(i)}\Big)^2\;dW^{(j)}(x)\ ,\\
  \mbox{and}\;\;\;\mu_K &\dff&
  \bbe[X_K]=\int_0^{\infty}x\;dW^K(x)\\
  &=&\sum_{j=1}^KQ_j
  \int_0^{\infty}x\;dW^{(j)}(x)=\sum_{j=1}^KQ_j\mu_{(j)}\ .
\end{eqnarray*}
As discussed in~\cite[Sec. 4.6]{Arnold:book, Sanayei:IT07}, for ordered exponentially distributed random variables $X_K$ we have
\begin{equation}
    \label{eq:sigma_n}
    \sigma^2_K<2+2\mu_{(1)}\Big(\mu_{(1)}-\mu_{K}\Big)\ ,
\end{equation}
\begin{equation}
    \label{eq:mu_1}
    \mbox{and}\;\;\;\log K+\zeta+\frac{1}{2(K+1)}\leq\mu_{(1)}\leq\log
    K+\zeta+\frac{1}{2K}\ ,
\end{equation}
where by noting that $\zeta\approx0.577$ is the Euler-Mascheroni constant we have $\mu_{(1)}\doteq \log K$. Also
\begin{equation*}
     \mu_{(1)}-\log\bigg(\sum_{j=1}^KjQ_j\bigg)-\zeta-0.5\leq\mu_K\leq\mu_{(1)}.
\end{equation*}
By taking into account the constraint in (\ref{eq:Q1}), as
$K\rightarrow\infty$
\begin{equation}
    \label{eq:mu_n}
    1-\frac{\log\bigg(\sum_{j=1}^KjQ_j\bigg)-\zeta-0.5}{\log K}\leq\mu_K\leq 1.
\end{equation}
Equations~(\ref{eq:mu_1})~and~(\ref{eq:mu_n}) together show that
\begin{equation}
    \label{eq:mu_n2}
    \mu_{(1)}\doteq\mu_K\doteq\log K,
\end{equation}
which also implies that $\mu_K\rightarrow\infty$. Taking into
account~(\ref{eq:sigma_n})~ and~(\ref{eq:mu_n2}) we also conclude that
$\lim_{K\rightarrow\infty}\frac{\sigma_K}{\mu_K}=0$ and therefore the
conditions of Theorem~\ref{th:limit} are satisfied. Hence, from
Theorem~\ref{th:limit}
\begin{eqnarray*}
    \int_0^{\infty}\log(1+ax)\;dW^K(x)&=&\bbe\Big[\log(1+a X_K)\Big] \doteq \log\Big(1+a\bbe[X_K]\Big)\\
     &=&\log\Big(1+a\mu_K\Big)\doteq\log\log K+\log(a)\ .
\end{eqnarray*}

\section{Proof of Lemma \ref{lemma:f}}
\label{app:lemma:f}
By the expansion of $\Big(x+(1-x)\Big)^K$ we have
\begin{align*}
    f(x,j)=& 1-\sum_{i=j+1}^K{K\choose i}x^{K-i}(1-x)^i= 1-\sum_{i=j+1}^K{K\choose K-i}x^{K-i}(1-x)^i\\
    &= 1-\sum_{i=0}^{K-(j+1)}{K\choose i}(1-x)^{K-i}x^i=1-f(1-x,K-j-1),
\end{align*}
where it can be concluded that $f'(x,i)=f'(1-x,K-j-1)$. So it is
sufficient to show that $f'(x,i)\geq 0$ for $x\leq \frac{1}{2}$ and
for all $j=1,\dots, K-1$. For this purpose we consider two cases of
$j\leq\lfloor\frac{K}{2}\rfloor$ and
$j>\lfloor\frac{K}{2}\rfloor$. \\
{\underline{\bf Case 1: $j\leq\lfloor\frac{K}{2}\rfloor$}}
\begin{eqnarray}
  \nonumber f'(x,j) &=& \sum_{i=0}^{j}{K\choose i}(K-i)x^{K-i-1}(1-x)^i-{K\choose
  i}ix^{K-i}(1-x)^{i-1}\\
  \label{eq:lemma:f1}&=& \sum_{i=0}^{j}{K\choose
  i}(K-i)x^{K-i-1}(1-x)^{i-1}\Big[K(1-x)-i\Big],
\end{eqnarray}
where it can be shown that for $x\leq \frac{1}{2}$
\begin{equation*}
    K(1-x)-i\geq K(1-x)-j\geq
    K(1-x)-\frac{K}{2}=\frac{K}{2}(1-2x)\geq 0.
\end{equation*}
{\underline{\bf Case 2: $j>\lfloor\frac{K}{2}\rfloor$}}\\
Define $a_i=1-\frac{1}{2}\delta(\lfloor\frac{K}{2}\rfloor-i)$, where
$\delta(\cdot)$ is the Dirac delta function. Therefore, we get
\begin{equation*}
   f(x,j) = f(x,K-j-1)+\sum_{i=K-j}^{\lfloor\frac{K}{2}\rfloor}a_i{K\choose
   i}\Big[x^{K-i}(1-x)^i+x^i(1-x)^{K-i}\Big].
\end{equation*}
For $x\leq \frac{1}{2}$ we get
\begin{eqnarray}
   \nonumber f'(x,j) &=& f'(x,K-j-1)\\
   \nonumber &+&\sum_{i=K-i}^{\lfloor\frac{K}{2}\rfloor}a_i{K\choose i}
   x^{K-i-1}(1-x)^{i-1}\Big[K-i-Kx\Big]+\underset{\geq x^{K-i-1}(1-x)^{i-1}}{\underbrace{x^{i-1}(1-x)^{K-i-1}}}\Big[i-Kx\Big]\\
   \label{eq:lemma:f2} &\geq& f'(x,\underset{\leq \lfloor\frac{K}{2}\rfloor}{\underbrace{K-j-1}}) + \sum_{i=K-i}^{\lfloor\frac{K}{2}\rfloor}a_i{K\choose
   i} x^{K-i-1}(1-x)^{i-1}\underset{\geq 0}{\underbrace{\Big[K-2Kx\Big]}}\geq 0
\end{eqnarray}
From~(\ref{eq:lemma:f1})~and~(\ref{eq:lemma:f2}) it is concluded
that for $x\leq \frac{1}{2}$, $f(x,j)$ is an increasing function of
$x$, which provides the proof.

\bibliographystyle{IEEEtran}
\bibliography{IEEEabrv,MIMO_Network}

\end{document}